\documentclass[12pt]{article}

\usepackage{newpxtext} 
\usepackage{newpxmath} 

\usepackage{graphicx} 
\usepackage{setspace}
\usepackage[rgb]{xcolor}
\usepackage{verbatim}
\usepackage{tabularx}
\usepackage{subcaption}
\usepackage{amsgen,amsmath,amstext,amsbsy,amsopn,tikz, bm, bbm}

\usepackage{amsthm}
\usepackage{fancyhdr}
\usepackage[colorlinks=true, urlcolor=blue,  linkcolor=blue, citecolor=blue]{hyperref}
\usepackage[colorinlistoftodos]{todonotes}
\usepackage{rotating}
\usepackage[shortlabels]{enumitem}

\usepackage{natbib}

\usepackage[capitalize, noabbrev]{cleveref}
\crefname{equation}{}{}

\usepackage{geometry}
\newtheorem{theorem}{Theorem}
\newtheorem{proposition}{Proposition}

\newtheorem{lemma}{Lemma}

\newtheorem{corollary}{Corollary}

\newtheorem{definition}{Definition}

\definecolor{darkgreen}{rgb}{0.0,0.65,0.45}

\title{Filling Positions Without Transfers: \linebreak Screening on Outside Options\thanks{
We are especially grateful to Piotr Dworczak, Paul Milgrom, Michael Ostrovsky, Ilya Segal, and Andy Skrzypacz for invaluable conversations. 
We thank Itai Ashlagi, Daniele Condorelli, Matthew Gentzkow, Klajdi Hoxha, Matthew Jackson, Ravi Jagadeesan, Zi Yang Kang, Marzena Rostek, Filip Tokarski, Joanna Tyrowicz, Mitchell Watt, Frank Yang, and seminar audiences at GRAPE IMD and Stanford for helpful comments. 
We thank many employees and returned volunteers at the U.S. Peace Corps for inspiration and useful discussions.}}

\author{
Morteza Honarvar\thanks{Graduate School of Business, Stanford University. Email: \texttt{honar@stanford.edu}} 
\and Joanna Krysta\thanks{Department of Economics, Stanford University. Email: \texttt{jkrysta@stanford.edu}} 
\and Eric Tang\thanks{Graduate School of Business, Stanford University. Email: \texttt{ertang@stanford.edu}}
}

\date{February 10, 2026}

\begin{document}

\maketitle
\begin{abstract}
    A designer offers vertically-differentiated positions to agents in the absence of transfers. 
    Agents have private outside options and may reject their offers ex-post. 
    The designer has preferences over the quantity of agents who accept each position. 
    We show that under a general condition on the distribution of outside options, an optimal mechanism for the designer offers all agents an identical lottery, and we characterize this mechanism. 
    When our condition does not hold, the optimal mechanism may require screening agents by offering a menu of distinct lotteries. 
    Our results follow from a decomposition of agents' participation probabilities in any feasible mechanism. 
    
    \bigskip 
    
\end{abstract}

\newpage

\section{Introduction}

Consider a designer who offers vertically differentiated positions to agents in the absence of transfers. 
Agents have private outside options and reject an offer if it is less attractive than their outside option. 
The designer has preferences over the quantity of agents who accept each position: 
for example, she may wish to fill as many positions as possible or to ensure a minimum quantity of agents in each position. 
The designer's inability to observe these outside options constrains her ability to implement her desired allocation. 

One example is the U.S. Peace Corps volunteer program, which assigns volunteers to development projects around the world, paying volunteers a small stipend. 
Some locations are more desirable and receive substantially more applications. 
Applicants also typically have outside options, which can make it challenging for the Peace Corps to fill all open positions.\footnote{The former director of the U.S. Peace Corps discussed the challenges of recruiting volunteers: 
``In the early days when Peace Corps began, it really was the only opportunity for people who wanted to work or live overseas. ... 
Our applicants have lots of choices now.'' (\cite{greenblatt2014response}, ``In Response To Dwindling Applications, Peace Corps Makes Big Changes''). More generally, many employers struggle to fill all open positions (\cite{canon2025firefighter}).}
The same forces can prevent other employers from filling positions: hospitals assign nurses to heterogeneous shifts; cities assign police officers to neighborhoods and assign teachers to schools. If the employer is a public institution or wages are standardized by union agreement, she may not be able to set distinct wages for different positions.\footnote{Government employees' wages are often determined by a fixed wage schedule, which limits the ability of a public employer to set distinct wages for different positions. In the private sector, employee unions may also negotiate a standardized wage across positions. As \cite{budish2011combinatorial} writes of the healthcare firm McKesson, ``McKesson’s nursing shift assignment software, eShift, has both a fixed-price version and an auction version, depending on whether the client hospital has discretion to use flexible wages (e.g., because of union restrictions).'' Even when the designer does not face such restrictions, our work characterizes how well the designer can meet her objectives without using additional transfers.} 
As a final example, consider a university department assigning students to different sections of its introductory course: 
the department may aim to maximize enrollment, but students prefer sections with high-quality instructors and vary in their outside options.

A central feature of our setting is that designers value utilization — the quantity of each position that is filled.
Prior work on the unit-assignment problem without transfers, beginning with \cite{hylland1979efficient}, typically focuses on Pareto efficiency or total welfare for agents, omitting the employer's own preferences.\footnote{
More recent work in this tradition includes \cite{bogomolnaia2001new}, \cite{miralles2012cardinal} and \cite{ortoleva2021cares}. \cite{krysta2014size} and \cite{bogomolnaia2015size} provide approximate bounds on assignment sizes, whereas our framework allows the designer to compute the exact optimum for general quantity objectives.} 
Indeed, in assignment problems without transfers, there is no revenue objective for the designer to maximize. 
Yet settings with outside options naturally suggest other important objectives for the designer: maximizing the quantity of employees hired or students enrolled.

To illustrate the designer's challenge, consider an employer with low- and high-quality positions.  
To maximize the number of positions filled, she might wish to assign low-quality positions to agents with low-quality outside options and high-quality positions to those with high-quality outside options. 
However, such a policy conflicts with agents' reporting incentives: 
agents with weak outside options misrepresent themselves to secure high-quality positions, leaving low-quality positions unfilled. 
The designer's problem is thus: 
given heterogeneity in agents' private outside options, how should she design her offers to maximize her objective? 

We formalize this problem in a large-market model where the designer allocates a mass of vertically differentiated positions to a mass of agents with private outside options. Vertically differentiated positions serve as a tractable benchmark, and may be a reasonable approximation in many settings of interest.\footnote{Applicants to the U.S. Peace Corps often rank Costa Rica, Thailand, and Morocco as the most desirable locations (\cite{damico2023competitive}). 
Students prefer course sections taught by better lecturers.
\cite{ba2021police} find that police officers prefer to be assigned to safer neighborhoods.}
We also model agents as free to quit: after the mechanism assigns positions to agents, each agent may still reject his assignment.\footnote{In the Peace Corps setting, an agent cannot be forced to accept an assignment, and volunteers frequently withdraw. This is analogous to the quitting rights of \cite{compte2009veto} and \cite{haberman2025auctions}.}
The designer's objective is a function of the mass of agents who accept each position.\footnote{
For example, a government agency may aim to maximize positions filled, subject to filling each type of position with at least 50\% of its capacity, or subject to the quantity of agents hired in each type of position being approximately equal.}
Among possible mechanisms, one particularly simple mechanism is to use a homogeneous lottery to randomly offer each agent a position. 
Agents will accept offer realizations whose quality exceeds that of their outside options. 
We call such a mechanism a \emph{common lottery}. 

The space of feasible mechanisms is considerably broader and more complex, however.
To screen agents with different outside options, the designer can offer a menu of different lotteries over positions.
Intuitively, a lottery that offers a low-quality position with high probability is appealing only to agents with low-quality outside options, who may prefer it over a lottery that offers a high-quality position with low probability.\footnote{In our course allocation example, students could have a choice between entering a lottery for the high-quality instructor's course, or entering a lottery for the low-quality instructor's course. 
In the U.S. Peace Corps, applicants may either apply to serve anywhere in the world or apply to serve in a particular country. 
Returned Peace Corps volunteers advise applicants, ``the more flexible you are in terms of placement, the better your chances are of getting selected. 
Don’t apply to Thailand, Morocco, or Costa Rica directly – there are so many applicants that choose to apply there. 
Be flexible!'' (\cite{damico2023competitive}).} 
Another possibility is the Competitive Equilibrium from Equal Incomes (CEEI) mechanism, introduced by \cite{hylland1979efficient}.\footnote{By an appropriate revelation principle (e.g., \cite{myerson1979incentive}), we allow for mechanisms replicating trading outcomes.}
This allows agents to trade probability shares among themselves: 
agents with different outside options will choose to buy different lotteries.

Determining the optimal mechanism in this space appears challenging. 
Indeed, in our setting, local incentive constraints (between adjacent types) are insufficient to guarantee the global incentive compatibility of a mechanism. 
This arises because agents make ex-post participation decisions that allow double deviations: 
they may misreport their type and then reject an undesirable offer. 
Consequently, classical mechanism design tools are inadequate for solving the designer's problem. 
Nevertheless, we characterize when the optimal mechanism takes a particularly simple form.

Our main result shows that under a general condition, an optimal mechanism for the designer is a common lottery. 
The sufficient condition is the convexity of $1/F$, where $F$ is the distribution function of agents' outside options.\footnote{
    This condition is satisfied if $F$ is log-concave, a property shared by many distributions. 
    In this sense, the convexity of $1/F$ is weaker than the regularity conditions typically assumed in mechanism design. 
}
Our proof uses an original constructive argument---which decomposes any feasible allocation into a common allocation and a component due to incentive constraints---to replace any feasible mechanism with a common lottery that preserves the quantity of agents filled in each position.
Consequently, our result holds for general designer objective functions.
Under strict convexity of $1/F$, when the designer's objective is strictly monotone in utilization, common lotteries are uniquely optimal. 
This implies that the designer should strictly prefer our optimal mechanism over alternatives such as CEEI.

When our condition holds, it is straightforward to determine the optimal mechanism. The problem reduces to a standard consumer choice problem: the designer chooses the mass of agents to allocate to each position, subject to a budget constraint in which $1/F(x_k)$ serves as the ``price'' of allocating an additional agent to position $x_k$.\footnote{
    $F(x_k)$ is the proportion of agents whose outside options are strictly worse than position $x_k$. 
} 
We provide an explicit characterization of this optimal mechanism. 
This mechanism can also be implemented through a ``capped random priority'' mechanism, a generalization of the random priority mechanism.\footnote{In the capped random priority mechanism, agents are randomly ordered, and each agent in turn selects his most preferred position among those remaining. By imposing caps on the availability of each position, the designer can replicate the allocation of the common lottery. We model random priority as in \cite{che2010asymptotic}, which also connects our implementation to the probablistic serial mechanism introduced by \cite{bogomolnaia2001new}.}

We further provide a partial converse. 
When our convexity condition does not hold, common lotteries may be strictly suboptimal. 
In this case, the optimal mechanism requires screening agents by offering a menu of distinct lotteries. 
Our proof uses a novel variational argument, in which we begin with the optimal common lottery and perturb the allocation to improve the objective by screening some agents with weak outside options.

The proof of our main result focuses on each agent’s total participation probability, that is, the probability that the agent accepts some offer. 
For any feasible mechanism, we decompose these participation probabilities into two components: a common component shared across agents, and a type-specific component. 
For each position, we construct the common component by taking the mean acceptance probability among agents who find that position acceptable.
Taken alone, this common component defines a common lottery that assigns the same quantity of agents to each position as the initial mechanism. 
Under our convexity condition, we show the type-specific component is positive, implying that the common component alone is a feasible mechanism. 
We further provide a duality interpretation of this decomposition.

The remainder of the paper is organized as follows. 
\cref{sec:related-lit} discusses related literature. 
\cref{sec:model} describes our model and provides examples of feasible mechanisms. 
\cref{sec:results} introduces our main result---a sufficient condition for the optimality of a common lottery---and a partial converse. 
\cref{sec:characteriz} describes features of any feasible mechanism.
\cref{sec:common-ordinal-preferences} extends our results to common ordinal preferences.
\cref{sec:future} concludes.

\subsection{Related Literature}\label{sec:related-lit}

Beginning with \cite{hylland1979efficient}, a wide literature considers the item-assignment problem without transfers, in which agents with preferences are matched to objects without preferences. 
This literature typically considers objectives of Pareto efficiency or total welfare, which differ from our designer's objective.
\cite{bogomolnaia2001new} study ordinal efficiency in the unit assignment problem, which does not encompass our designer's objectives. 
\cite{budish2013designing} generalize the CEEI mechanism of \cite{hylland1979efficient} to multi-unit assignment settings. 
\cite{miralles2012cardinal} and \cite{ortoleva2021cares} analyze  welfare-maximizing mechanisms without transfers for allocating two objects and vertically differentiated goods, respectively. 
\cite{ashlagi2016optimal} consider more general designer objectives, but impose Pareto efficiency constraints that we do not consider.

We depart from this literature by modeling the designer's preferences over which items (positions) are ultimately allocated, which is distinct from a welfare objective.
In particular, when agents have outside options, the designer's choice of mechanism may affect which positions are ultimately filled. In contrast with
\cite{bogomolnaia2015size} and \cite{krysta2014size}, who provide worst-case guarantees for assignments of maximal size, we (i) characterize the exact optimal mechanism for a class of assignment problems and (ii) consider cardinal utilities, enabling the designer to screen agents by their relative preference intensity.\footnote{In \cite{bogomolnaia2015size} and \cite{krysta2014size}, truthfulness is defined with respect to ordinal comparisons of utility. We instead consider cardinal utility functions; our introduction and examples (e.g., \cref{fig:alternative-examples} and \cref{fig:optmech-example-binary}) describe how a designer may use a menu of lotteries to screen agents on their outside option. In these menus, distinct lotteries yield distinct expected utilities for each agent.}\textsuperscript{,}\footnote{In addition, we characterize the optimal mechanism when the designer has quantity objectives distinct from simply maximizing the total mass of agents assigned, and our partial converse characterizes settings in which a probabilistic serial mechanism is strictly suboptimal. \cite{krysta2014size} allow welfare to depend on which agents are matched to which positions, a dependence that is absent in our formulation.} The vertical ordering of positions in our model is similar to \cite{cres2001scheduling}, who compare the performance of two fixed mechanisms; we obtain an optimal mechanism in the entire space of feasible mechanisms.

The two-sided matching literature, beginning with \cite{gale1962college}, considers the preferences and strategic behavior of each individual position, which is distinct from considering the objectives of an organization (employer) with preferences over the entire assignment profile. 
The stability constraints in the two-sided matching literature also differ from those in our settings of employee assignment.
This also connects us to the broader school choice literature: for example, \cite{akbarpour2022centralized} focus on the equity effects of different mechanisms in the presence of outside options. 
Literature such as \cite{cowgill2024matching} consider matching employees to teams within a firm using transfers.

Finally, another broad distinction with these works is that we model agents' ex-post participation decisions, which arise naturally in many settings. 
\cite{compte2007quitting} and \cite{haberman2025auctions} study ex-post participation in bargaining and auction settings, respectively.
In our assignment setting, employees may leave their job after being assigned to a given team, and students may drop their class after being assigned to a given instructor. 
This gives rise to the possibility of double deviations, which affects the designer's optimal mechanism.

\section{Model}\label{sec:model}
\subsection{Framework}
We consider a model with a continuum of agents and a finite type space, and a continuum of positions with a finite set of possible qualities.
We normalize the mass of available positions to $1$ and let $D$ denote the mass of agents. 
$\Theta$ denotes the space of types, which represent agents' outside options. 
The distribution of agents' types is given by $f$, a full-support probability mass function on $\Theta$, and $F$ denotes the associated cumulative distribution function. 
$X$ denotes the space of position qualities, and $g$ is their probability mass function; $g(x_k)$ represents the available capacity of positions with quality $x_k$. 
In our baseline model, we assume the type and the position spaces are evenly-spaced grids of $N$ distinct types, with $X = \Theta = \left\{0,\frac{1}{N-1}, \dots, \frac{N-2}{N-1},1\right\}$.\footnote{
    One interpretation of the assumption $X = \Theta$ is that agents’ outside options take values in a rich type space $\Theta$. 
    When this space is sufficiently fine, each position quality $x_k$ can be approximated by a point in $\Theta$. 
    Note that we allow $g(x_k)=0$ for some $x_k$, so that not every grid point in $X$ must correspond to a position with positive supply. 
}
Let $x_k \coloneq \frac{k}{N-1}$ denote the quality of the position with index $k$, and let $\theta_i \coloneq \frac{i}{N-1}$ denote the outside option with index $i$.

In our baseline model, we assume that agents have homogeneous cardinal preferences: 
an agent who accepts position $x$ obtains utility $x$.\footnote{
Our results extend if agents' utility functions are $w(x)$ for some increasing utility function $w$. 
We may simply normalize the problem by considering the cumulative distribution function of positions $y$ to be $F(w^{-1}(y))$, and taking agent utility to be $u(y;\theta) = \max\{w(y),w(\theta)\}$. 
\cref{sec:common-ordinal-preferences} describes this normalization and extends our results to common ordinal preferences.} 
Each agent has an ex-post participation decision. 
If assigned a position $x \geq \theta$ better than his outside option, the agent accepts. 
If assigned a position $x < \theta$ worse than his outside option, the agent instead consumes his outside option $\theta$. 
Agents' utility is therefore given by $u(x; \theta) = \max\{x,\theta\}$.
An agent's utility with type $\theta_i$ from any assignment lottery $c\in\Delta(X\cup\{\varnothing\})$ is given by:
\[\sum_{k=0}^{N-1}\max\{x_k;\theta_i\}c(x_k)+\theta_i\left(1-\sum_{k=0}^{N-1}c(x_k)\right)=\theta_i+\sum_{k=i}^{N-1}(x_k-\theta_i)c(x_k).\]

We allow the designer's objective to depend in a general way on the ultimate mass of agents who accept each position. 
For a given mechanism, let $s_k$ denote the mass of agents who accept position $x_k$ in the mechanism. 
The designer's objective function is $V:\mathbb{R}_+^N \to \mathbb{R}$, with $V(\bm{s})$ representing her value from filling positions with masses vector $\bm{s}$. 
We impose that $V$ be upper semicontinuous and place no further restrictions on $V$.

A particularly important case of the designer's objective is when the designer solely values maximizing the quantity of positions filled (maximizing utilization). 
We denote this objective function by $V_{\text{fill}}$, with $V_{\text{fill}}(\bm{s}) \coloneq \sum_{k=0}^{N-1} s_k$. 
The designer chooses a stochastic mechanism without transfers to maximize her objective function.

\subsection{Direct Mechanisms}

By the revelation principle, it is without loss of generality to restrict our attention to direct mechanisms. 
In the absence of transfers, a direct mechanism is a map $a: \Theta \to \Delta(X \cup \{\varnothing\})$, where $\varnothing$ denotes receiving no assignment from the mechanism (and hence consuming one's outside option). 
We let $a(x;\theta)$ denote the probability that an agent of type $\theta$ is offered position $x$. 
The direct mechanism's incentive compatibility (IC) constraints ensure that an agent does not find it profitable to misreport his type, then accept only allocated positions above his true outside option. 
The mechanism's individual rationality (IR) constraint guarantees that agents are not assigned a position worse than their outside option. 
The mass of agents who accept position $k$ under direct mechanism $a$ is given by $s_k(a) = D \sum_{i=0}^k a(x_k;\theta_i) f(\theta_i)$.\footnote{
Note that given the ex-post IR constraint, this is equivalent to defining $s_k(a)$ by $D \sum_{i=0}^{N-1} a(x_k;\theta_i) f(\theta_i)$ for each $k$. 
Indeed, defining the designer's objective in the way we have, the ex-post IR constraints will never constrain the designer in equilibrium. 
The designer would never strictly benefit from offering an agent of type $\theta$ a position of quality lower than $\theta$: 
it does not increase the designer's objective, and can only tighten the designer's other constraints. 
This means we may drop the Ex-Post IR constraints from the analysis if desired, though we include them here for exposition.
}

The designer's problem is therefore
$$
\max_{a: \Theta \to \Delta(X \cup \{\varnothing\})} V(\bm{s}(a)), 
$$
subject to constraints:
\begin{align}
    \sum_{k=i}^{N-1} (x_k - \theta_i)\, a(x_k;\theta_i) 
    & \geq \sum_{k=i}^{N-1} (x_k - \theta_i)\, a(x_k;\theta_j),
    \qquad \forall i,j 
    \tag{$\text{IC}_{i,j}$} \\[6pt]
    D \sum_{i=0}^{N-1} a(x_k;\theta_i) f(\theta_i) 
    & \leq g(x_k),
    \qquad \forall k 
    \tag{Position-Feasibility} \\[6pt]
    \sum_{k=0}^{N-1} a(x_k;\theta_i)\, 
    & \leq 1,
    \qquad \forall i
    \tag{Agent-Feasibility} \\[6pt]
    a(x_k;\theta_i) & = 0,
    \qquad \forall k<i
    \tag{Ex-Post IR}
\end{align}

We say a direct mechanism $a$ is \textit{feasible} if, in addition to being incentive compatible, it is ex-post individually rational and satisfies the position- and agent-feasibility constraints. 
We say a direct mechanism $a$ is \emph{optimal} for objective $V$ (and model primitives $F$ and $D$) if $a$ attains the maximum value of $V$ in the set of feasible mechanisms. 

\subsection{Common Lotteries}

One class of mechanisms to consider is what we call a \emph{common lottery}. 
Common lotteries are particularly simple for a mechanism designer to implement: 
they require offering only a single homogeneous lottery to all agents. 
Any common lottery has an indirect and a direct representation. 
We may describe a common-lottery indirect mechanism as a distribution $\tilde{a} \in \Delta(X \cup \{ \varnothing \})$. 
In the common-lottery indirect mechanism, each agent receives an offer $x$ with probability $\tilde{a}(x)$, then simply accepts the realized outcome if it is above their outside option. 
We can write any indirect common lottery $\tilde{a}(\cdot)$ as an outcome-equivalent direct mechanism $a: 
\Theta \to \Delta(X \cup \{\varnothing\})$, where we truncate the lottery for each type at his outside option.
\begin{definition}
    Direct mechanism $a$ is a \emph{common lottery} if there exists $\tilde{a} \in \Delta(X \cup \{ \varnothing \})$ such that
    $$a(x_k;\theta_i) = 
    \begin{cases}
        \tilde{a}(x_k), & \theta_i \leq x_k \\
        0 , & \theta_i > x_k.
    \end{cases}$$
\end{definition}

Many feasible mechanisms are not common lotteries: 
the designer may screen different types by offering a menu of different lotteries. 

If the designer knows that her optimum can be attained with a common lottery, it is simple to compute an optimal mechanism: 
she simply computes the single lottery that maximizes her objective, subject to the agents accepting each offered position. 
Otherwise, the optimal mechanism may take a complex form that involves screening different agents.

\subsection{Discussion of Model}

In allowing the designer's objective function to depend flexibly on the mass allocated to each position, we allow the designer to entertain a wide range of objective functions. In the example of a university department assigning students to sections, the department may wish to maximize enrollment, or bring enrollment in each section as close to some target percentage as possible, or maximize enrollment subject to maintaining approximately equal enrollment in every section. In the employment example, the employer may place particular weight on filling certain types of positions. 

The agent's choice to consume his outside option instead of his assignment can be viewed as an ex-post participation decision. 
A prospective employee can reject an offer or move to a new job if his placement is insufficiently desirable. 
A student can choose to drop the course section he has been assigned. Our IC constraints incorporate the possibility of double deviations, in which the agent misreports his type, but continues to accept only the realized offers that are above his outside option.

We further assume that agents share the same ordinal preferences and vary only in their outside options: students may all prefer to be placed into classes with better instructors, and volunteers may generally prefer particular locations. Finally, our large-market model with a continuum of agents can be a reasonable approximation when an organization assigns many agents to positions.\footnote{The asymptotic results in \cite{che2010asymptotic} show that the outcomes of large finite markets approach the outcome of a similar continuum assignment model.}

\subsection{Mechanism Examples}\label{sec:mechanism-examples}

We now provide some examples of potential mechanisms.
In these examples, suppose that the designer wishes to maximize the total number of positions filled, so her objective function is $V_{\text{fill}}$. 
Suppose $D=1$, $N = 4$, and types and positions are distributed uniformly. 
In the following diagrams, we take the x-axis to be types $\theta$ and the y-axis to be the positions $x$. 
Each cell $(\theta_i,x_k)$ represents the probability $a(x_k;\theta_i)$. 
As described earlier, an agent of type $\theta_i$ rejects any offered position $x_k < \theta_i$, and the designer obtains no benefit from offering these positions, so we omit these cells. 

If agents' outside options were not private information, the designer's first-best allocation could exactly fill all positions by assigning all agents to exactly their outside option, as in \cref{fig:first-best-allocation-example}. 
Agents are each assigned to a position yielding utility equal to their outside option. 
This allocation clearly violates incentive constraints: an agent with a low outside option has an incentive to claim he has a high outside option, in order to obtain a higher-quality assignment.
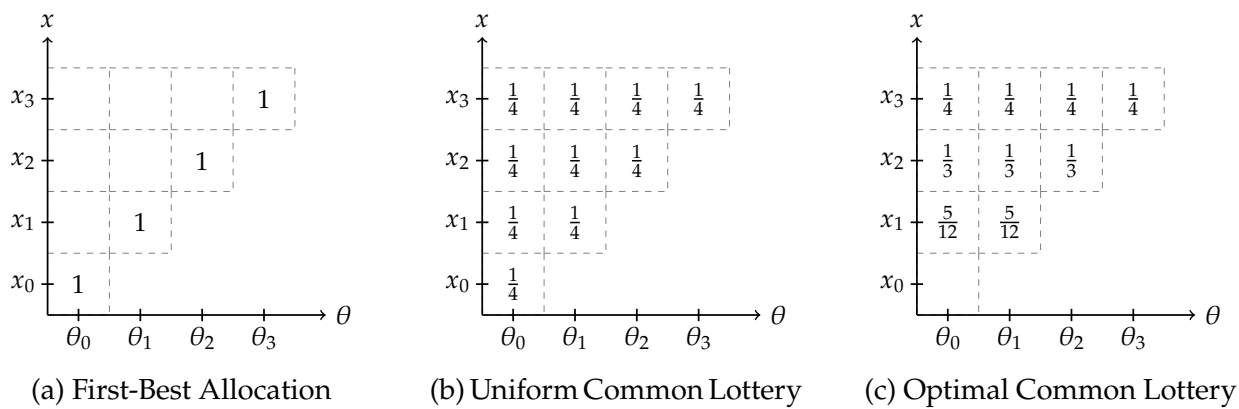
\begin{figure}[htbp]
    \centering
    \begin{subfigure}{0.3\textwidth}
        \centering
        \resizebox{\textwidth}{!}{
        \begin{tikzpicture}
            
            \foreach \x in {0,...,4} {
                \foreach \y in {0,...,4} {
                    \ifnum\y>\x
                        
                        \draw[very thin, dashed, gray] (\x,\y-1) -- (\x,\y);
                        
                        \draw[very thin, dashed, gray] (\x,\y) -- (\x+1,\y);
                    \fi
                    \ifnum\y=\x
                        \ifnum\x<4
                        
                        \draw[very thin, dashed, gray] (\x,\y) -- (\x+1,\y);
                        \draw[very thin, dashed, gray] (\x+1,\y+1) -- (\x+1,\y);
                        \fi
                    \fi
                }
            }

            \draw[->, thick] (0,0) -- (4.5,0) node[right] {$\theta$};
            \draw[->, thick] (0,0) -- (0,4.5) node[above] {$x$};

            \foreach \t in {0,...,3} {
                \draw[thick] (\t+0.5,-0.1) -- (\t+0.5,0.1) node[below=4pt] {$\theta_{\t}$};
            }

            \foreach \t in {0,...,3} {
                \draw[thick] (-0.1,\t+0.5) -- (0.1,\t+0.5) node[left=4pt] {$x_{\t}$};
            }

            \foreach \i in {0,...,3} {
                \node at (\i+0.5,\i+0.5) {1};
            }
        \end{tikzpicture}
        }
        \caption{First-Best Allocation}
        \label{fig:first-best-allocation-example}
    \end{subfigure}
    \hfill
    \begin{subfigure}{0.3\textwidth}
        \centering
        \resizebox{\textwidth}{!}{
        \begin{tikzpicture}
            
            \foreach \x in {0,...,4} {
                \foreach \y in {0,...,4} {
                    \ifnum\y>\x
                        
                        \draw[very thin, dashed, gray] (\x,\y-1) -- (\x,\y);
                        
                        \draw[very thin, dashed, gray] (\x,\y) -- (\x+1,\y);
                    \fi
                    \ifnum\y=\x
                        \ifnum\x<4
                            
                            \draw[very thin, dashed, gray] (\x,\y) -- (\x+1,\y);
                            \draw[very thin, dashed, gray] (\x+1,\y+1) -- (\x+1,\y);
                        \fi
                    \fi
                }
            }

            \draw[->, thick] (0,0) -- (4.5,0) node[right] {$\theta$};
            \draw[->, thick] (0,0) -- (0,4.5) node[above] {$x$};

            \foreach \t in {0,...,3} {
                \draw[thick] (\t+0.5,-0.1) -- (\t+0.5,0.1) node[below=4pt] {$\theta_{\t}$};
            }

            \foreach \t in {0,...,3} {
                \draw[thick] (-0.1,\t+0.5) -- (0.1,\t+0.5) node[left=4pt] {$x_{\t}$};
            }

            \foreach \x in {0,...,3} {
                \foreach \y in {\x,...,3} {
                    \node at (\x+0.5,\y+0.5) {$\tfrac{1}{4}$};
                }
            }
        \end{tikzpicture}
        }
        \caption{Uniform Common Lottery}
        \label{fig:uniform-lottery-allocation}
    \end{subfigure}
    \hfill
    \begin{subfigure}{0.3\textwidth}
        \centering
        \resizebox{\textwidth}{!}{
        \begin{tikzpicture}
            
            \foreach \x in {0,...,4} {
                \foreach \y in {0,...,4} {
                    \ifnum\y>\x
                        
                        \draw[very thin, dashed, gray] (\x,\y-1) -- (\x,\y);
                        
                        \draw[very thin, dashed, gray] (\x,\y) -- (\x+1,\y);
                    \fi
                    \ifnum\y=\x
                        \ifnum\x<4
                            
                            \draw[very thin, dashed, gray] (\x,\y) -- (\x+1,\y);
                            \draw[very thin, dashed, gray] (\x+1,\y+1) -- (\x+1,\y);
                        \fi
                    \fi
                }
            }

            \draw[->, thick] (0,0) -- (4.5,0) node[right] {$\theta$};
            \draw[->, thick] (0,0) -- (0,4.5) node[above] {$x$};

            \foreach \t in {0,...,3} {
                \draw[thick] (\t+0.5,-0.1) -- (\t+0.5,0.1) node[below=4pt] {$\theta_{\t}$};
            }

            \foreach \t in {0,...,3} {
                \draw[thick] (-0.1,\t+0.5) -- (0.1,\t+0.5) node[left=4pt] {$x_{\t}$};
            }

            \foreach \x in {0,...,3} {
                \node at (\x+0.5,3.5) {$\tfrac{1}{4}$};
            }

            \foreach \x in {0,...,2} {
                \node at (\x+0.5,2.5) {$\tfrac{1}{3}$};
            }

            \foreach \x in {0,...,1} {
                \node at (\x+0.5,1.5) {$\tfrac{5}{12}$};
            }
        \end{tikzpicture}
        }
        \caption{Optimal Common Lottery}
        \label{fig:optimal-common-lottery-allocation}
    \end{subfigure}
    \caption{First-best and common lottery examples}
    \label{fig:mechanism-examples}
\end{figure}

To satisfy the incentive constraints, the designer could allocate positions with a uniform lottery, as in \cref{fig:uniform-lottery-allocation}, in which she offers each position with equal probability to each type. 
Any agent offered a position worse than his outside option will reject the offer. 
This mechanism allocates a total mass of $5/8$. 

With a little more thought, the designer may realize that to maximize the number of positions filled, she should be offering agents the high-quality positions less often than she offers the low-quality positions, as these positions are the ones agents are most likely to accept. 
She can implement this using the common lottery given in \cref{fig:optimal-common-lottery-allocation}, in which she fully fills the highest positions and leaves the lower positions only partly filled. 

Again, agents reject any offered position below their outside option. 
This mechanism allocates a total mass of $17/24 > 5/8$. 
Within the class of common lottery mechanisms, this is the designer's optimal mechanism.

However, the designer can implement mechanisms that are more complex than common lotteries.
Concretely, she can screen agents by their outside options, by offering agents a menu of distinct lotteries and letting them choose their desired lotteries. 
The designer might hope to fill even more positions by doing so.

For example, in \cref{fig:screening-with-menus}, the mechanism offers agents two distinct lotteries. 
The first lottery offers agents a higher probability of being allocated some position, but only allocates lower-quality positions. 
The second lottery offers agents a lower probability of being allocated some position, but allocates them higher-quality positions. 
An agent with a low outside option prefers the ``safer'' first lottery. An agent with a high outside option places relatively low value (or zero value) on the lower-quality positions, and prefers the ``riskier'' second lottery. This mechanism therefore screens agents into different lotteries by their outside option.
The designer may hope that this allows her to fill more low-quality positions, thereby improving her objective. 
Indeed, she is now able to completely fill position $x_1$. 
However, her overall objective value is still $5/8$. 
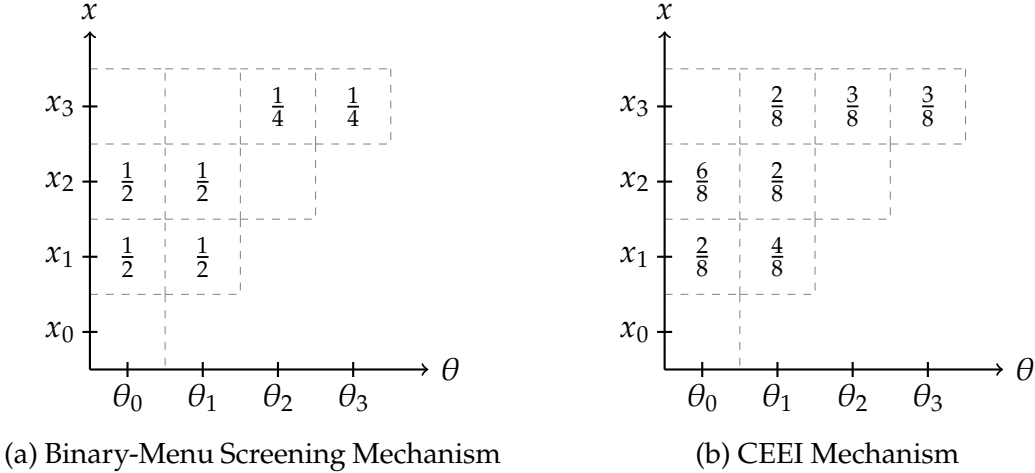
\begin{figure}[htbp]
    \centering
    
    \begin{subfigure}{0.4\textwidth}
        \centering
        \begin{tikzpicture}
        
        \foreach \x in {0,...,4} {
            \foreach \y in {0,...,4} {
            \ifnum\y>\x
                
                \draw[very thin, dashed, gray] (\x,\y-1) -- (\x,\y);
                
                \draw[very thin, dashed, gray] (\x,\y) -- (\x+1,\y);
            \fi
            \ifnum\y=\x
                \ifnum\x<4
                    
                    \draw[very thin, dashed, gray] (\x,\y) -- (\x+1,\y);
                    \draw[very thin, dashed, gray] (\x+1,\y+1) -- (\x+1,\y);
                \fi
            \fi
            }
        }
        
        \draw[->, thick] (0,0) -- (4.5,0) node[right] {$\theta$};
        \draw[->, thick] (0,0) -- (0,4.5) node[above] {$x$};
        
        \foreach \t in {0,...,3} {
            \draw[thick] (\t+0.5,-0.1) -- (\t+0.5,0.1) node[below=4pt] {$\theta_{\t}$};
        }
        
        \foreach \t in {0,...,3} {
            \draw[thick] (-0.1,\t+0.5) -- (0.1,\t+0.5) node[left=4pt] {$x_{\t}$};
        }
        
            \foreach \x in {2,...,3} {
                \node at (\x+0.5,3.5) {$\tfrac{1}{4}$};
            }

            \foreach \x in {0,...,1} {
                \node at (\x+0.5,2.5) {$\tfrac{1}{2}$};
            }

            \foreach \x in {0,...,1} {
                \node at (\x+0.5,1.5) {$\tfrac{1}{2}$};
            }
        \end{tikzpicture}
        \caption{Binary-Menu Screening Mechanism}
        \label{fig:screening-with-menus}
    \end{subfigure}
    \hspace{0.05\textwidth}
    \begin{subfigure}{0.4\textwidth}
        \centering
        \begin{tikzpicture}
            
            \foreach \x in {0,...,4} {
                \foreach \y in {0,...,4} {
                \ifnum\y>\x
                    
                    \draw[very thin, dashed, gray] (\x,\y-1) -- (\x,\y);
                    
                    \draw[very thin, dashed, gray] (\x,\y) -- (\x+1,\y);
                \fi
                \ifnum\y=\x
                    \ifnum\x<4
                        
                        \draw[very thin, dashed, gray] (\x,\y) -- (\x+1,\y);
                        \draw[very thin, dashed, gray] (\x+1,\y+1) -- (\x+1,\y);
                    \fi
                \fi
                }
            }

            \draw[->, thick] (0,0) -- (4.5,0) node[right] {$\theta$};
            \draw[->, thick] (0,0) -- (0,4.5) node[above] {$x$};

            \foreach \t in {0,...,3} {
                \draw[thick] (\t+0.5,-0.1) -- (\t+0.5,0.1) node[below=4pt] {$\theta_{\t}$};
            }

            \foreach \t in {0,...,3} {
                \draw[thick] (-0.1,\t+0.5) -- (0.1,\t+0.5) node[left=4pt] {$x_{\t}$};
            }
            
            \node at (1+0.5,3.5) {$\tfrac{2}{8}$};
            \node at (2+0.5,3.5) {$\tfrac{3}{8}$};
            \node at (3+0.5,3.5) {$\tfrac{3}{8}$};
            
            \node at (0+0.5,2.5) {$\tfrac{6}{8}$};
            \node at (1+0.5,2.5) {$\tfrac{2}{8}$};
            
            \node at (0+0.5,1.5) {$\tfrac{2}{8}$};
            \node at (1+0.5,1.5) {$\tfrac{4}{8}$};
        \end{tikzpicture}
        \caption{CEEI Mechanism} 
        \label{fig:hylland-zeckhauser-allocation}
    \end{subfigure}
    
    \caption{Screening mechanism examples}
    \label{fig:alternative-examples}
\end{figure}
An alternative method to screen agents is the competitive equilibrium from equal incomes (CEEI) mechanism, as in \cite{hylland1979efficient}, in which agents can freely trade probability shares in each position.\footnote{The common lottery does not correspond to a CEEI allocation, as the second-highest type would spend his entire budget on shares of the highest position, hence must have a higher probability of obtaining this position than other agents have.}
In this example, one equilibrium of CEEI is given in \cref{fig:hylland-zeckhauser-allocation}, and obtains objective $11/16$.\footnote{This equilibrium is supported by prices $p_0=p_1=0, p_2=1,p_3=2$, with each agent having budget of $b=3/4$. 
As in CEEI, each agent purchases a bundle that maximizes his utility, subject to his total probability of allocation being $1$. In Hylland and Zeckhauser's specification of CEEI, each agent is required to purchase the \emph{cheapest} feasible bundle that maximizes his utility, which guarantees Pareto Efficiency for agents; we drop this requirement, which allows the designer to fill more positions. However, the designer still strictly prefers the optimal common lottery to this CEEI allocation.} In both cases, the designer attains an objective value strictly less than that of the optimal common lottery in \cref{fig:optimal-common-lottery-allocation}.

In this setting, these examples suggest that the designer cannot do better than a common lottery, which is also a particularly simple format of mechanism to implement. Yet there is a wide space of feasible mechanisms available to the designer.
When is it the case that the designer cannot do better than a common lottery?

\section{Results}\label{sec:results}

We now introduce and discuss our main result: a sufficient condition for the optimality of a common lottery. When a common lottery is optimal, the designer need not screen agents by their outside options. This makes it simple to characterize the designer's optimal mechanism. We also provide a partial converse to our result, showing that when our condition is violated, the designer may benefit from offering agents a menu of lotteries.

The key technical difficulty in our analysis is that local IC constraints are insufficient to characterize the set of feasible mechanisms. 
To see this, consider the mechanism in \cref{fig:local-ICs-insufficient}. 
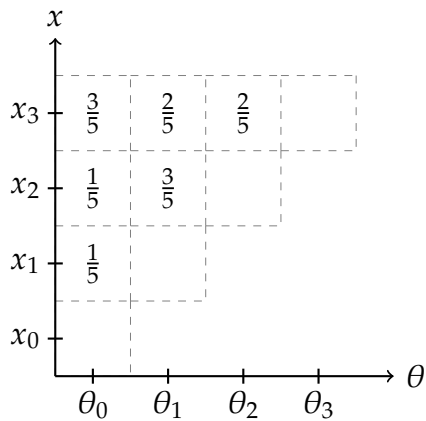
\begin{figure}[h]
\centering
\begin{tikzpicture}
  
  \foreach \x in {0,...,4} {
    \foreach \y in {0,...,4} {
      \ifnum\y>\x
        
        \draw[very thin, dashed, gray] (\x,\y-1) -- (\x,\y);
        
        \draw[very thin, dashed, gray] (\x,\y) -- (\x+1,\y);
      \fi
      \ifnum\y=\x
        \ifnum\x<4
            
            \draw[very thin, dashed, gray] (\x,\y) -- (\x+1,\y);
            \draw[very thin, dashed, gray] (\x+1,\y+1) -- (\x+1,\y);
        \fi
      \fi
    }
  }

  \draw[->, thick] (0,0) -- (4.5,0) node[right] {$\theta$};
  \draw[->, thick] (0,0) -- (0,4.5) node[above] {$x$};

  \foreach \t in {0,...,3} {
    \draw[thick] (\t+0.5,-0.1) -- (\t+0.5,0.1) node[below=4pt] {$\theta_{\t}$};
  }

  \foreach \t in {0,...,3} {
    \draw[thick] (-0.1,\t+0.5) -- (0.1,\t+0.5) node[left=4pt] {$x_{\t}$};
  }

    \node at (3+0.5,3.5) {};

    \node at (2.5,2+0.5) {};
    \node at (2.5,3+0.5) {$\frac{2}{5}$};

    \node at (1.5,1+0.5) {};
    \node at (1.5,2+0.5) {$\frac{3}{5}$};
    \node at (1.5,3+0.5) {$\frac{2}{5}$};

    \node at (0.5,0+0.5) {};
    \node at (0.5,1+0.5) {$\frac{1}{5}$};
    \node at (0.5,2+0.5) {$\frac{1}{5}$};
    \node at (0.5,3+0.5) {$\frac{3}{5}$};
\end{tikzpicture}
\caption{Insufficiency of local IC constraints for feasibility.} 
\label{fig:local-ICs-insufficient}
\end{figure}
All local IC constraints are respected under this mechanism, but the global downward IC constraint $\text{IC}_{20}$ is not satisfied. 

We will show that, indeed, at the optimal mechanism when our condition is satisfied, all global downward IC constraints bind, as well as the local IC constraints. In order to first present our main result, we defer further discussion of these binding constraints and the insufficiency of local IC constraints to \cref{sec:characteriz}. We establish our contribution via two results. Despite the technical challenge of binding global downward ICs, we prove our main result by an original direct constructive argument which transforms a feasible direct mechanism into a common lottery that preserves the same objective. We prove our partial converse with a variational argument that perturbs any common lottery to improve the designer's objective, if our convexity condition does not hold.

\subsection{Optimality of Common Lotteries}
Suppose the designer wishes to allocate a mass $s_k$ of agents to position $x_k$ in a common-lottery mechanism. Then the designer must offer position $x_k$ to a mass of $\frac{s_k}{F(x_k)}$ agents, because an agent accepts the offered position exactly when his outside option is below $x_k$. Indeed, the expression $1/F$ plays an important role in determining the optimal mechanism: our sufficient condition for the optimality of common lotteries is the convexity of $1/F$, appropriately defined on the discrete domain $\Theta$. 

A function $\varphi:\Theta \to \mathbb{R}$ is convex on $\Theta$ if for all $\alpha \in (0,1)$ and all $\theta,\theta'\in \Theta$ such that $\alpha\theta+(1-\alpha)\theta' \in \Theta$, we have $\varphi(\alpha\theta+(1-\alpha)\theta') \leq \alpha \varphi(\theta) + (1-\alpha) \varphi(\theta')$. Throughout, we use convexity to refer to convexity on $\Theta$. In a continuous type space, for $1/F$ to be convex, the rate of increase in $\ln f$ must be at most twice the rate of increase in $\ln F$, where $f$ denotes the probability density function associated with $F$.\footnote{
For convexity, we require $\frac{d^2}{d\theta^2}\frac{1}{F(\theta)} = \frac{2f^2(\theta) - F(\theta)f'(\theta)}{F^3(\theta)} \geq 0$, which is equivalent to 
$$2\frac{f(\theta)}{F(\theta)} \geq \frac{f'(\theta)}{f(\theta)} \quad \Leftrightarrow \quad 2(\ln F(\theta))' \geq (\ln f(\theta))'.$$
} 
Note that a non-increasing $f$ is sufficient for this condition to hold. 
More interestingly, $f$ being log-concave is also sufficient.\footnote{
Log-concavity of $f$ implies log-concavity of $F$. 
$F$ being log-concave means $\frac{d^2}{d\theta^2}\log F(\theta) = \frac{f'(\theta)F(\theta) - f^2(\theta)}{F^2(\theta)} \leq 0$, or, equivalent, $ (\ln F(\theta))' \geq (\ln f(\theta))'$. 
}
These in turn imply that many common distributions, such as the uniform distribution, normal distribution, logistic distribution, Pareto distribution, and exponential distribution, satisfy our condition.\footnote{
To be precise, because we work in a finite-type model, the discrete approximation of these distributions satisfy our condition.
Note also that in our model the support of the distribution is $[0, 1]$; nevertheless, truncating a function preserves its log-concavity.} 
\cref{thm:implementable-by-common-lottery} shows that in these common settings, a designer can use a common lottery to maximize her objective.

\begin{theorem}\label{thm:implementable-by-common-lottery}
    Suppose $1/F$ is convex. For any feasible direct mechanism $a$, there exists a feasible common lottery $\tilde{a}(\cdot)$ with $\bm{s}(a)=\bm{s}(\tilde{a})$. In particular, for any designer objective function $V$ and agent mass $D$, there exists an optimal mechanism that is a common lottery.
\end{theorem}

This result states that regardless of the designer's objective, she can do no better than to offer every agent the same common lottery. In this setting, screening agents by offering a menu of lotteries does not benefit the designer. We also need not impose a monotonicity assumption on $V$.\footnote{The upper semi-continuity of $V$ guarantees existence of an optimal mechanism. Even in the absence of this condition, screening offers no additional benefit for the designer.}

In the proof, we begin with an arbitrary feasible mechanism and construct a common lottery that preserves the designer’s objective and all feasibility constraints. This implies that the designer's optimum can be attained by a common lottery. To perform this transformation, for each position $x_k$, we “equalize’’ the lottery by assigning every agent the same probability of receiving $x_k$, while preserving the initial total mass of agents allocated to that position. This transformation automatically preserves position feasibility (because the same mass of agents is assigned to each position as before) and, because the outcome is a common lottery, it satisfies incentive compatibility.

The remaining issue is agent feasibility, that is, that no agent is assigned with probability exceeding one. To establish this, we identify the key tradeoff in the environment: using type-dependent lotteries changes the need to allocate to the lowest type. When $1/F$ is convex, distorting allocations across types is costly in terms of incentive constraints. This tradeoff can be expressed through an allocation decomposition.\footnote{Similar expressions hold for other types with altered coefficients.} The probability of allocating $\theta_0$ to a position in any feasible allocation can be written as the sum of (i) a common component obtained via a degenerate mean-preserving contraction of the original allocation and (ii) a type-dependent component that is a linear combination of the IC inequalities:\footnote{ Here $(\text{IC}_{i,j})$ refers to the expression obtained by taking all terms in the inequality $\text{IC}_{i,j}$ to the left-hand side, and multiplying by $N-1$. By our constraints, this expression must be greater than zero.}

\[P(\theta_0)\coloneq\sum_{k=0}^{N-1}a(x_k;\theta_0)=\overbrace{\sum_{k=0}^{N-1}\sum_{i=0}^{N-1}\frac{a(x_k;\theta_i)f(\theta_i)}{F(\theta_k)}}^{\text{common allocation}}+\overbrace{\sum_{i=0}^{N-2}\lambda_{i,i+1}(\text{IC}_{i,i+1})+\sum_{i=1}^{N-1}\sum_{j=0}^{i-1}\lambda_{i,j}(\text{IC}_{i,j})}^{\text{value of information}}\]
The latter term captures precisely the incentive cost of procuring private information through differentiated lotteries.

Monotonicity of allocation in type (see \Cref{prop:MON}) implies that the lowest type receives the highest allocation probability and therefore constitutes the bottleneck agent: relaxing type dependence primarily operates by reducing this type’s allocation odds. The decomposition formalizes this intuition. It is a structural property of the IC constraints in our model and holds throughout the feasible set. Taking an appropriate weighted sum of the IC inequalities—where the multipliers $\lambda_{i,j}$ are positive exactly\footnote{The multipliers $\lambda_{i,i+1}$ on local upward incentive constraints are always weakly positive (while non-local upward incentive constraints are $0$). The non-local downward incentive constraints $\lambda_{i,j}=f(\theta_j)\left((\frac{1}{F(\theta_{i+1})}-\frac{1}{F(\theta_i)})-(\frac{1}{F(\theta_i)}-\frac{1}{F(\theta_{i-1})})\right),$ hence convexity of $1/F$ implies that these multipliers are also all positive.} when $1/F$ is convex—implies the agent-feasibility constraint for the common lottery. When $1/F$ is strictly convex, the same argument yields uniqueness of the optimal common lottery (see \cref{common-lotto-uniquely-opt}).

Conceptually, this argument reduces the question to a comparison between improved allocation success and the shadow cost of tightening IC constraints. The coefficients in the decomposition admit a formal interpretation as shadow prices on the IC constraints, though this interpretation is not needed for the proof itself.\footnote{To obtain the partial converse when convexity of $1/F$ does not hold at some position $x_k$, we note that the information decomposition provides a recipe for reverse-engineering a utilization-preserving perturbation that will strictly improve upon utilization of an optimal common lottery. The novel variational perturbation we find moves mass in proportion to their respective shadow prices.}

In the proof described above, the origins of the multipliers $\lambda_{\text{IC}_{i,j}}$ and the $1/F$ convexity condition may not appear obvious. To illustrate how these multipliers arise, we also provide a proof of \cref{thm:implementable-by-common-lottery} using duality, in \cref{proof:alternative-implementable-by-common-lotto}. This proof considers the following related problem. Suppose a designer may choose both $a$ and $D$, but is given a vector of target position masses $\bm{s} \in \mathbb{R}^N$ and is constrained to allocate $s_k$ agents to each position $x_k$, subject to all feasibility constraints. The designer's objective is to minimize the mass of applicants $D$. If for any $\bm{s}$, the designer can attain her optimum with a common lottery, then in the original problem, the designer can always attain her optimum with a common lottery. This is because in the original problem, for any direct mechanism $a$ that is not a common lottery and allocates position masses $\bm{s}$, the designer can transform it into a feasible common lottery $\tilde{a}$ that allocates the same masses $\bm{s}$ to each position, and hence attains the same objective $V(\bm{s})$. The shadow prices $\{\lambda_{\text{IC}_{i,j}}\}$ in this duality proof become the multipliers used in our direct constructive proof above. The construction of these shadow prices also motivates the perturbation proof of our partial converse, \cref{thm:partial-converse}.

\subsection{Partial Converse}\label{subsubsec:partial-converse-result}
Here we prove our partial converse, showing that if our convexity condition does not hold and $V$ is strictly increasing, there exists a mass of agents $D$ for which no common lottery is optimal.
To provide an example, let $D=1$, $N=3$, and let the objective be $V_{\text{fill}}$. Suppose positions are distributed uniformly and the type distribution is $f(\theta_0)=\frac{1}{3},f(\theta_1)=\frac{1}{12},f(\theta_2)=\frac{7}{12}$. Consider the two mechanisms in \Cref{fig:optmech-example-main}:

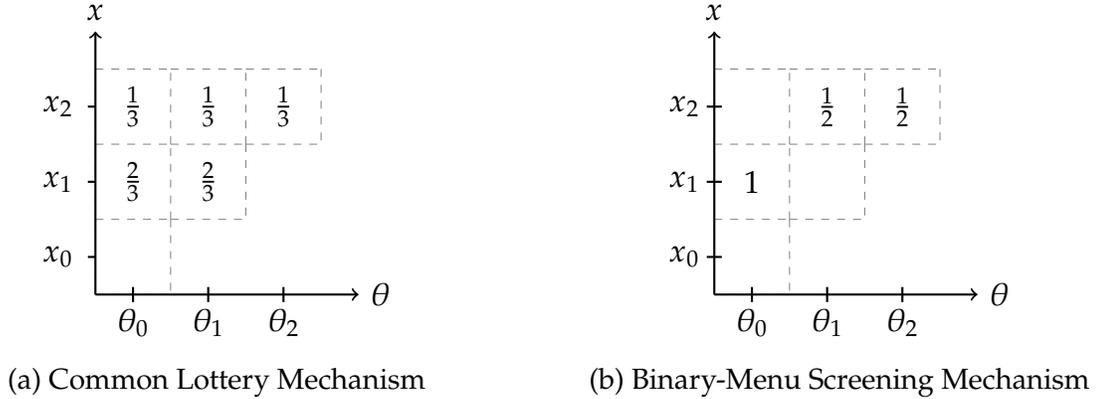
\begin{figure}[htbp]
    \centering
    
    \begin{subfigure}{0.5\textwidth}
        \centering
        \begin{tikzpicture}
            
            \foreach \x in {0,...,3} {
                \foreach \y in {0,...,3} {
                \ifnum\y>\x
                    
                    \draw[very thin, dashed, gray] (\x,\y-1) -- (\x,\y);
                    
                    \draw[very thin, dashed, gray] (\x,\y) -- (\x+1,\y);
                \fi
                \ifnum\y=\x
                    \ifnum\x<3
                        
                        \draw[very thin, dashed, gray] (\x,\y) -- (\x+1,\y);
                        \draw[very thin, dashed, gray] (\x+1,\y+1) -- (\x+1,\y);
                    \fi
                \fi
                }
            }

            \draw[->, thick] (0,0) -- (3.5,0) node[right] {$\theta$};
            \draw[->, thick] (0,0) -- (0,3.5) node[above] {$x$};

            \foreach \t in {0,...,2} {
                \draw[thick] (\t+0.5,-0.1) -- (\t+0.5,0.1) node[below=4pt] {$\theta_{\t}$};
            }

            \foreach \t in {0,...,2} {
                \draw[thick] (-0.1,\t+0.5) -- (0,\t+0.5) node[left=4pt] {$x_{\t}$};
            }

            \node at (0+0.5,2.5) {$\tfrac{1}{3}$};
            \node at (1+0.5,2.5) {$\tfrac{1}{3}$};
             \node at (2+0.5,2.5) {$\tfrac{1}{3}$};

            \node at (0+0.5,1.5) {$\frac{2}{3}$};
            \node at (1+0.5,1.5) {$\frac{2}{3}$};
            
        \end{tikzpicture}
        \caption{Common Lottery Mechanism}
        \label{fig:optmech-example-commonlotto} 
    \end{subfigure}
    \begin{subfigure}{0.49\textwidth}
        \centering
        \begin{tikzpicture}
        
        \foreach \x in {0,...,3} {
            \foreach \y in {0,...,3} {
            \ifnum\y>\x
                
                \draw[very thin, dashed, gray] (\x,\y-1) -- (\x,\y);
                
                \draw[very thin, dashed, gray] (\x,\y) -- (\x+1,\y);
            \fi
            \ifnum\y=\x
                \ifnum\x<3
                    
                    \draw[very thin, dashed, gray] (\x,\y) -- (\x+1,\y);
                    \draw[very thin, dashed, gray] (\x+1,\y+1) -- (\x+1,\y);
                \fi
            \fi
            }
        }

        \draw[->, thick] (0,0) -- (3.5,0) node[right] {$\theta$};
        \draw[->, thick] (0,0) -- (0,3.5) node[above] {$x$};

        \foreach \t in {0,...,2} {
            \draw[thick] (\t+0.5,-0.1) -- (\t+0.5,0.1) node[below=4pt] {$\theta_{\t}$};
        }

        \foreach \t in {0,...,2} {
            \draw[thick] (-0.1,\t+0.5) -- (0.1,\t+0.5) node[left=4pt] {$x_{\t}$};
        }

            \foreach \x in {1,...,2} {
                \node at (\x+0.5,2.5) {$\tfrac{1}{2}$};
            }

            \foreach \x in {0,...,0} {
                \node at (\x+0.5,1.5) {$1$};
            }
        \end{tikzpicture}
        \caption{Binary-Menu Screening Mechanism}
        \label{fig:optmech-example-binary} 
    \end{subfigure}

    \caption{Optimal Mechanism Example}
    \label{fig:optmech-example-main}
\end{figure}
The common lottery in \cref{fig:optmech-example-commonlotto} is the optimal common lottery in this setting: it fills position $x_2$ to capacity, then fills the maximum feasible mass of position $x_1$, attaining an objective of $\frac{11}{18}$. However, the binary-menu screening mechanism in \cref{fig:optmech-example-binary} is incentive compatible and achieves a higher objective of $\frac{2}{3}>\frac{11}{18}$. This turns out to be a general consequence when $1/F$ is not convex.

\begin{theorem}\label{thm:partial-converse}
    Suppose $1/F$ is not convex and $g$ has full support. 
    If the designer objective $V$ is strictly increasing, there exists $D$ such that no common lottery is optimal.

\end{theorem}

The intuition for this result, and the $1/F$ convexity condition, is as follows.\footnote{We also believe this applies to more general objective functions $V$, as long as $V$ is non-satiated at the optimum, but our proof does not yet formally include this.} We use a variational argument, beginning with a common lottery and finding an improvement. Imagine that the designer begins by using a common lottery to allocate positions, and obtains allocated position masses $\bm{s} \in \mathbb{R}^{N}$ allocated to the positions. The feasibility constraint, given by the total mass of agents needed in the population to attain this allocation mass is $\sum_{k=0}^{N-1} \frac{1}{F(x_k)} s_k \leq D$. This is because in a common lottery, an agent offered position $x_k$ accepts the position with probability $F(x_k)$. The importance of convexity should now be apparent. If $1/F$ is not convex, then there exists a position mass $\tilde{\bm{s}}$ that is a mean-preserving spread of $\bm{s}$ and would result in a lower left-hand side. That is, if the designer was willing to allocate agents in a mean-preserving spread over positions, she buys herself slack in this feasibility inequality. That in turn leaves an additional mass of agents to allocate to other positions.

If the designer actually chose the lottery to obtain the mean-preserving spread allocation $\tilde{\bm{s}}$, this could violate the supply constraints or lower her objective. However, the designer can still obtain allocation $\bm{s}$ while lowering the mass of agents needed. She can do so by additionally imposing a mean-preserving \emph{contraction} on the positions allocated to \emph{low} types $\theta$. Indeed, suppose we choose some type $\theta$, consider their lottery allocation $a(\cdot;\theta)$, and change this lottery to a mean-preserving contraction. We show in our proof that this mean-preserving contraction satisfies all incentive-compatibility constraints.\footnote{If this contracts mass between positions $x_k$ and $x_{k'}$, we refer to the interval $[x_k,x_{k'}]$ as the \emph{contraction region}. The contraction satisfies incentive compatibility constraints because (i) any agent with type above the contraction region sees the change as irrelevant (ii) any agent with type below the contraction region (including type $\theta$) sees the same expected value before and after for $a(\cdot;\theta)$ and (iii) any agent with type in the contraction region now views the lottery allocated to $\theta$ as worse than before.}

The designer may perform these two steps in a way that preserves the position allocation $\bm{s}$ while reducing the mass of agents needed to obtain the allocation. This leaves an additional mass of agents that can be assigned to strictly improve the objective. \cref{fig:perturbation-partial-converse} displays the two steps of the perturbation, with darker shades representing higher allocation probabilities. Our proof in the appendix formalizes this variational argument.

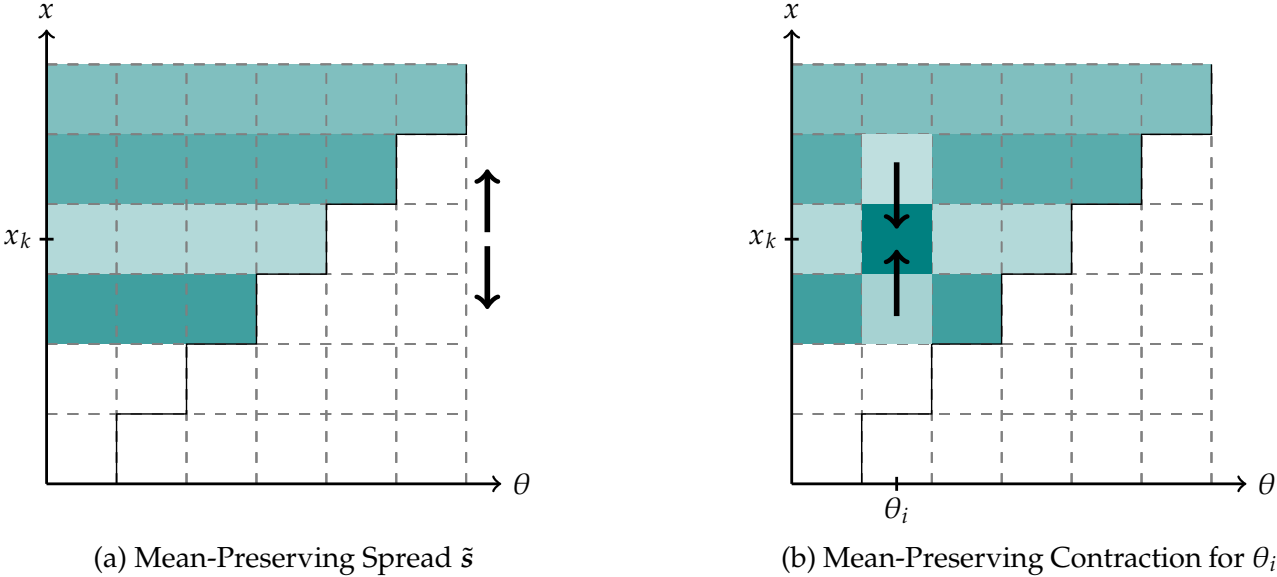
\begin{figure}[htbp]
    \centering
    
    \begin{subfigure}{0.4\textwidth}
        \centering
        \resizebox{\textwidth}{!}{
        \begin{tikzpicture}[scale=0.7, transform shape]
            \useasboundingbox (0,-0.6) rectangle (6.8,6.8);
            \def\shadearray{{0,0,0,75,30,65,50}}
            \foreach \x in {0,...,6} {
                \foreach \y in {1,...,6} {
                    \ifnum\x<6
                        \ifnum\y>\x
                            \pgfmathsetmacro{\shadeval}{\shadearray[\y]}
                            \fill[teal!\shadeval] (\x,\y-1) rectangle (\x+1,\y);
                        \fi
                    \fi
                    \draw[thin, dashed, gray] (\y,0) -- (\y,6);
                    \draw[thin, dashed, gray] (0,\y) -- (6,\y);
                }
            }

            \draw[black]
                (0,0)--(1,0)--(1,1)--(2,1)--(2,2)--(3,2)--(3,3)
                --(4,3)--(4,4)--(5,4)--(5,5)--(6,5)--(6,6);

            \draw[->, thick] (0,0) -- (6.5,0) node[right] {$\theta$};
            \draw[->, thick] (0,0) -- (0,6.5) node[above] {$x$};

            \draw[thick] (-0.1,3.5) -- (0.1,3.5) node[left=4pt] {$x_k$};

            \draw[->, ultra thick] (6.3,3.6) -- (6.3,4.5);
            \draw[->, ultra thick] (6.3,3.4) -- (6.3,2.5);
        \end{tikzpicture}
        }
        \caption{Mean-Preserving Spread $\tilde{\bm{s}}$}
        \label{fig:perturbation-partial-converse-a}
    \end{subfigure}
    \hfill
    \begin{subfigure}{0.4\textwidth}
        \centering
        \resizebox{\textwidth}{!}{
        \begin{tikzpicture}[scale=0.7, transform shape]
            \useasboundingbox (0,-0.6) rectangle (6.8,6.8);
            \def\shadearray{{0,0,0,75,30,65,50}}
            \foreach \x in {0,...,6} {
                \foreach \y in {1,...,6} {
                    \ifnum\x<6
                        \ifnum\y>\x
                            \pgfmathsetmacro{\shadeval}{\shadearray[\y]}
                            \fill[teal!\shadeval] (\x,\y-1) rectangle (\x+1,\y);
                        \fi
                    \fi
                    \draw[thin, dashed, gray] (\y,0) -- (\y,6);
                    \draw[thin, dashed, gray] (0,\y) -- (6,\y);
                }
            }

            \draw[black] (0,0) -- (1,0) -- (1,1) -- (2,1) -- (2,2) -- (3,2) -- (3,3) -- (4,3) -- (4,4) -- (5,4) -- (5,5) -- (6,5) -- (6,6);

            \draw[->, thick] (0,0) -- (6.5,0) node[right] {$\theta$};
            \draw[->, thick] (0,0) -- (0,6.5) node[above] {$x$};

            \draw[thick] (1.5,-0.1) -- (1.5,0.1) node[below=4pt] {$\theta_i$};

            \draw[thick] (-0.1,3.5) -- (0.1,3.5) node[left=4pt] {$x_k$};

            \fill[teal!100] (1,3) rectangle (2,4);
            \fill[teal!35] (1,2) rectangle (2,3);
            \fill[teal!25] (1,4) rectangle (2,5);
            
            \draw[->, ultra thick] (1.5,4.6) -- (1.5,3.65);
            
            \draw[->, ultra thick] (1.5,2.4) -- (1.5,3.35);
        \end{tikzpicture}
        }
        \caption{Mean-Preserving Contraction for $\theta_i$}
        \label{fig:perturbation-partial-converse-b}
    \end{subfigure}

    \caption{Perturbation for Partial Converse}
    \label{fig:perturbation-partial-converse}
\end{figure}

As alluded to in the previous subsection, this perturbation also illustrates the intuition behind the shadow price on downward multipliers $\lambda_{\text{IC}_{i,j}}$ in the proof of \cref{thm:implementable-by-common-lottery}. In the case when $1/F$ is convex, the designer can reduce the mass of agents necessary by allocating a mean-preserving \emph{contraction} $\tilde{\bm{s}}$ of the target allocation $\bm{s}$. To obtain the actual target allocation $\bm{s}$, the designer would have to allocate to some type $\theta_j$ a mean-preserving spread of her original allocation lottery. However, such a mean-preserving spread will violate IC constraints: if the mean-preserving spread shifts mass from $a(x_k;\theta_j)$ to $a(x_{k+1};\theta_j)$, then it will violate $\text{IC}_{k,j}$. This is why such a perturbation is not feasible to improve the designer's objective when $1/F$ is convex. When the IC constraint $\text{IC}_{k,j}$ is relaxed, this perturbation becomes possible: the designer's gain is the shadow price in the duality proof of \cref{thm:implementable-by-common-lottery}.

\subsection{Optimal Mechanisms}\label{subsec:optimal-common-lotteries}

If $1/F$ is convex, \cref{thm:implementable-by-common-lottery} shows that the optimal mechanism takes a simple form, which also allows us to characterize the optimal mechanism. The designer's optimal mechanism reduces to a classical consumer choice problem, with a simple budget set constraint. Given that the optimum is attainable with a common lottery, if the designer wishes to fill position $x_k$ with a mass $s_k$ of agents, then she must offer position $x_k$ to a mass $\frac{s_k}{F(x_k)}$ of agents. There is a total mass $D$ of agents, which serves as the designer's budget constraint.

\begin{proposition}\label{prop:optimal-common-lottery}
    Suppose $1/F$ is convex. An optimal mechanism for the designer is the common lottery $a$ with $\bm{s}(a) \in \mathbb{R}^N_+$ satisfying:

    $$
    \max_{\bm{s}\in \mathbb{R}^N_+} V(\bm{s})
    $$

    subject to
    \begin{align*}
    \sum_{k=0}^{N-1} \frac{1}{F(x_k)} s_k & \leq D \\[6pt]
    s_k & \leq g(x_k) \qquad \forall \; k.
    \end{align*}
\end{proposition}

This formulation also makes clear that $\frac{1}{F(x_k)}$ functions as the price to fill an additional unit of position $x_k$. Given this result, we can explicitly characterize the designer's optimal mechanism when her objective function is $V_{\text{fill}}$. The designer obtains equal utility from allocating an additional unit of agents to each position $x_k$, up until the position capacity $g(x_k)$. Hence her optimal mechanism will consist of filling positions to capacity in ascending order of their price $\frac{1}{F(x_k)}$, that is, filling all positions with the highest qualities. The designer may fill these positions with type $\theta_i$ agents until the available mass of type $\theta_i$ is exhausted.

\begin{corollary}\label{cor:optimal-common-lottery-vfill}
    Suppose $1/F$ is convex. 
    Let $q_k \coloneq \frac{1}{D} \frac{g(x_{k})}{F(x_{k})}$, and $Q_k \coloneq \min\left\{1, \sum_{k'=k+1}^{N-1} q_{k'}\right\}$. 
    An optimal direct mechanism $a(\cdot; \cdot)$ for designer objective $V_{\text{fill}}$ satisfies
    \begin{equation*}
        a(x_k;\theta_i) = \begin{cases}
            \min\left\{q_k , 1-Q_k \right\} & \text{ if } \theta_i \leq x_k \\
            0 & \text{ if } \theta_i > x_k.
        \end{cases}
    \end{equation*}
\end{corollary}

For other objective functions $V$, the designer's problem is still a straightforward solution to the problem in \cref{prop:optimal-common-lottery}: she chooses $\bm{s} \in \mathbb{R}_+^N$ to maximize $V(\bm{s})$ subject to the budget constraint. In \cref{app:general-designer-objectives} we characterize the optimal mechanism for other objective functions $V$. As one example, if the designer's objective function $V$ places high value on filling the low-quality positions, her optimal mechanism may fill these positions before completely filling the high-quality positions. If her objective function $V$ is concave, her optimal mechanism may partially fill all positions. Regardless of her objective function, however, when our convexity condition holds, \cref{thm:implementable-by-common-lottery} shows that she can always implement this optimal mechanism with a common lottery.

\subsection{Implementation of Common Lotteries}

Having established that a common lottery is optimal under our general condition, and described the optimal common lottery, we now show that any common lottery can be implemented as a \emph{capped random priority} mechanism, a generalization of the random priority mechanism (also commonly referred to as random serial dictatorship). 

In the random priority mechanism, agents are randomly ordered, and each agent in turn selects his preferred position among the remaining positions. To see the intuitive connection between random priority and common lotteries, consider an agent participating in the random priority mechanism. When this agent is called to choose his position, this agent will claim the highest-quality remaining position $x_k$, so long as it is above his outside option $\theta$. Each agent essentially faces a lottery (which determines the highest available position when he is called to choose), and accepts his offer if that lottery offers him a position above his outside option. Because agents are randomly ordered, this lottery is identical for all agents.

More generally, any common lottery can be implemented by a mechanism that preserves the structure of the random priority mechanism, while potentially imposing further caps on the capacity of each position. Formally defining this mechanism in the continuum setting requires care; we use the same techniques as \cite{che2010asymptotic} to define this mechanism. Our formal definition of the capped random priority mechanism with capping vector $\mathbf{s}$, denoted $\text{CRP}(\mathbf{s})$, is in \cref{app:implementation-common-lotteries}. The cap on each position serves as the available quota of each position. As we show in \cref{prop:capped-rsd}, any feasible common lottery can be implemented by an outcome-equivalent CRP mechanism.

\begin{proposition}
\label{prop:capped-rsd}
    Fix a feasible common lottery $a$ and resulting allocation vector $\bm{s}(a)$. In $\text{CRP}(\mathbf{s}(a))$, the probability that an agent with outside option $\theta_i$ receives position $x_k$ is $a(x_k; \theta_i)$. 
\end{proposition}

As discussed earlier, in our setting with vertical preferences, one implementation of a common lottery is to draw a random position for each agent and allow them to accept or reject the position. \cref{prop:capped-rsd} shows that a capped random priority mechanism also implements a common lottery. As \cite{che2010asymptotic} show, in a continuum environment, random priority is equivalent to the probabilistic serial mechanism, first introduced by \cite{bogomolnaia2001new}: a designer could therefore similarly implement a common lottery using a ``Capped Probabilistic Serial'' mechanism.\footnote{\cite{che2010asymptotic} also show that under certain conditions, the outcome of a random priority mechanism in finite markets converges, as the markets grow in size, to the outcome of the random priority mechanism in this continuum model. This suggests that these implementation results may accurately inform designers' decisions in a large finite market.}

The particular role of caps in this implementation is that they allow the designer to steer agents to particular positions in the common lottery. As shown in \cref{cor:optimal-common-lottery-vfill}, if the designer's objective is simply $V_{\text{fill}}$, she completely fills any high-quality position before assigning agents to lower-quality positions. This corresponds to a case in which the designer uses a random priority mechanism with no caps. However, if the designer has objectives other than $V_{\text{fill}}$, she may assign agents to lower-quality positions even if a high-quality position $x_k$ is not filled to capacity $g(x_k)$. In such cases, the classic random priority mechanism would not attain the designer's optimal common lottery: agents would claim all available capacity of the high-quality position $x_k$, and may not claim the lower-quality position. However, the designer can still restore her optimal common lottery by using caps: she imposes a quota $s_k$ on the amount of position $x_k$ available to agents.

For objective $V_{\text{fill}}$, this implementation is also ``detail free'' in the sense that the designer only needs to know that $1/F$ is convex. In this case, for any such distribution, the optimal common lottery runs a random priority mechanism without any caps.

Finally, if $1/F$ is not convex and no common lottery is optimal, no capped random priority mechanism can attain the designer's optimum. For example, in \cref{fig:optmech-example-binary}, the designer cannot implement the optimal menu with a capped random priority mechanism, regardless of the caps that she uses. To see this, note that the optimal mechanism assigns agents of types $\theta_1, \theta_2$ to the highest-quality position with positive probability, but assigns $\theta_0$ to the highest-quality position with probability $0$. In any capped random priority mechanism, type $\theta_0$ would have positive probability of drawing a low priority number, and hence to obtain the highest-quality position. Our results therefore offer a sharp frontier for when the designer can and cannot obtain her optimum with a random priority mechanism.

\section{Characterizing Feasible Mechanisms}\label{sec:characteriz}
In this section we elaborate on general properties of feasible mechanisms. 
We first derive a necessary condition for an IC mechanism: monotonicity of allocation probability. We then characterize the sets of redundant and binding IC constraints.

We first show that IC constraints imply monotonicity of the probability of allocation: agents with lower-quality outside options must be allocated to some position with higher probability than agents with higher-quality outside options. For a direct mechanism $a$, let $P(\theta_i) \coloneq \sum_{k=i}^{N-1} a(x_k;\theta_i)$ denote the total probability that type $\theta_i$ is allocated any position.

\begin{proposition}[MON]\label{prop:MON}
    In any feasible mechanism, $P(\theta_i)$ is non-increasing. 
\end{proposition}

As we showed in \cref{fig:local-ICs-insufficient}, local IC constraints do not imply global IC constraints.
We first show that local IC constraints being satisfied do imply that global \emph{upward} IC constraints are satisfied. Intuitively, local IC constraints imply monotonicity, which implies that higher-type agents are allocated with lower probability, which is in turn sufficient to guarantee that low-type agents do not deviate upward. 

\begin{proposition}\label{prop:redundant-ICs}
The following incentive compatibility constraints are redundant:
\[\text{IC}_{N-1,j}, \forall_{j=0,1,...,N-2},\]
\[\text{IC}_{i,j}, \forall_{i}\forall_{j=i+2,i+3,...,N-1}.\]
\end{proposition}

However, unlike in typical mechanism design settings, local IC constraints do not imply that global \emph{downward} IC constraints are satisfied. Intuitively, high-type agents may still perform a ``double deviation'' to report a lower type, then reject any offers which are lower than his actual type. We show that these global downward constraints bind at the optimal mechanism when $1/F$ is convex. To state this result, we say that \textit{filling all positions is infeasible} if no feasible mechanism can fill all positions to capacity: that is, no feasible direct mechanism $a$ satisfies $\bm{s}_k(a) = g(x_k)$ for all $k$.

\begin{proposition}\label{prop:binding-ICs}
Suppose $1/F$ is convex, $V$ is strictly increasing, and filling all positions is infeasible. Then all local upward IC constraints bind at the optimum. If $1/F$ is \textit{strictly} convex, 
then all IC constraints not listed in \cref{prop:redundant-ICs} (all downward and local upward) bind at the optimum.    
\end{proposition}

The intuition for \cref{prop:binding-ICs} is that if these constraints did not bind, then as in the proof of \cref{thm:implementable-by-common-lottery}, we could transform the mechanism into a common lottery and be left with an additional mass of agents, which can be used to strictly improve the objective. \cref{prop:binding-ICs} in turn implies that common lotteries are uniquely optimal.  

\begin{corollary}\label{common-lotto-uniquely-opt}
 Suppose $1/F$ is strictly convex, $V$ is strictly increasing, and filling all positions is infeasible. Any optimal direct mechanism is a common lottery.
 
\end{corollary}

For objective $V_{\text{fill}}$, the unique optimal common lottery is given by \cref{cor:optimal-common-lottery-vfill}. Hence if $1/F$ is strictly convex and filling all positions is infeasible, the mechanism given by \cref{cor:optimal-common-lottery-vfill} is the unique optimal direct mechanism.

\section{Common Ordinal Preferences}\label{sec:common-ordinal-preferences}

In some settings, agents may share a common ordinal ranking of positions, but differ in their cardinal intensities for different positions. For example, in the setting of volunteer assignment, volunteers may agree that some positions are more desirable, but differ in the intensity of their preferences for those positions.\footnote{Similarly, in the assignment of teachers to schools, teachers may agree that certain schools are preferable, but differ in their intensity of this preference.} It may seem that a designer could improve her objective by screening agents on the intensity of this preference. However, our results extend with surprising generality: under an independence assumption on the distribution of agents' outside options, our convexity condition still guarantees that a designer cannot do better than a common lottery. As before, this allows the designer to determine the optimal mechanism.

To model these preferences, we disentangle the quality of positions from the utility that agents obtain from each position. Let $Q = \{q_0,q_1, \ldots, q_{N-1} \} \subset \mathbb{R}$ be the set of available position qualities, with $q_0 < q_1 < \ldots < q_{N-1}$. Type $\gamma \in \Gamma$ parametrizes an agent's utility function, and for simplicity we assume $\Gamma$ is finite. Type $v \in Q$ denotes the quality of an agent's outside option. Both $v$ and $\gamma$ are privately known by the agent. The utility an agent of type $\gamma$ obtains from consuming a position of quality $q_k$ is $u(q_k;\gamma)$. Similarly, the agent's utility from consuming his outside option $v$ is $u(v;\gamma)$. We continue to assume that positions are vertically ordered, that is, $u(\cdot;\gamma)$ is strictly increasing for all $\gamma \in \Gamma$. The introduction of $\gamma$ allows agents to vary in the intensity of their preference for higher-quality preferences.

In our baseline model, we assumed $\Gamma = \{\gamma^0\}$ is a singleton, so all agents shared identical cardinal preferences. In the baseline model, defining $x_k = u(q_k;\gamma^0)$ corresponds to the utility that an agent obtains from consuming a position of quality $q_k$, and defining $\theta_k = u(v;\gamma^0)$ corresponds to an agent's utility from consuming his outside option. Let $H: Q \to [0,1]$ denote the distribution function of the quality of agents' outside options. Then defining $F(\cdot) = H(u^{-1}(\cdot;\gamma^0))$, we have that $F$ is the distribution function for the utility of agent's outside options. This allowed us to work entirely in utility space with $x_k$, $\theta_k$ and $F$. Restated in this model, we obtain a result analogous to \cref{thm:implementable-by-common-lottery}.\footnote{To be precise, our baseline model also considered an evenly-spaced grid of cardinal utility, so $u(q_k;\gamma^0) = \frac{k}{N-1}$. \cref{prop:common-lottery-optimal-uneven-grid} shows that the result of \cref{thm:implementable-by-common-lottery} generalizes to unevenly spaced grids. In the proof of \cref{prop:common-lottery-optimal-uneven-grid}, we find new Lagrange multipliers on the incentive-compatibility constraints, constructed with the same intuition as before, to show that the result still holds. Our proof is nearly identical to that of our main theorem, but with generalized multipliers on our IC constraints. The necessary condition is that $\frac{1}{H(u^{-1}(\cdot; \gamma^0))}$ is convex, because $H(u^{-1}(x_k; \gamma^0))$ plays the same role as $F$ in our original proof: it is the mass of agents who would accept an offer of a position of utility $x_k$.} 

\begin{proposition}\label{prop:common-lottery-optimal-uneven-grid}
  Suppose $\Gamma = \{\gamma^0\}$ and $\frac{1}{H(u^{-1}(\cdot; \gamma^0))}$ is convex on $u(Q; \gamma^0)$. For any feasible direct mechanism $a$, there exists a feasible common lottery $\tilde{a}$ with $\bm{s}(a)=\bm{s}(\tilde{a})$. In particular, for any designer objective function $V$ and agent mass $D$, there exists an optimal mechanism that is a common lottery.
\end{proposition}

When agents have common ordinal preferences but heterogeneous cardinal preferences, $\Gamma$ is no longer a singleton. To gain traction in this setting, we assume that $\gamma$ and $v$ are independently distributed. The substance of this assumption is that the quality of an agent's outside option is not informative about the relative intensity of his preferences for different qualities. In our leading example of volunteer assignment, the quality of an applicant's outside job or graduate school offer is independent of his intensity of preferences for different volunteering locations. This may be a reasonable approximation in many settings; in this example, a candidate's external job offers may largely depend on exogenous factors.\footnote{In our example of assigning students to course sections, students' outside options are the other courses that are available to take. These are plausibly exogenous to their preference over courses.} We continue to use $H$ to denote the distribution function of agents' outside options: under our assumption, this distribution does not depend on $\gamma$.

In this enriched model, a direct mechanism is a map $a: Q \times \Gamma \to \Delta(Q \cup \{ \varnothing \})$. Let $a(q_k; v, \gamma)$ be the probability that direct mechanism $a$ allocates a position of quality $q_k$ to an agent with outside option of quality $v$ and type $\gamma$.\footnote{The constraints that a direct mechanism must satisfy are analogous to those in our baseline model; we describe these in \cref{app:common-ordinal-preferences}.} Direct mechanism $a$ is a \emph{common lottery} if there exists $\tilde{a} \in \Delta(Q \cup \{\varnothing\})$ such that $a(q_k;v_i, \gamma) = \tilde{a}(q_k) \cdot \mathbbm{1}[v_i \leq q_k]$ for all $v_i, q_k \in Q$ and all $\gamma \in \Gamma$.\footnote{That is, under a common lottery, the lottery offered to each agent does not depend on $\gamma$.} As in our baseline model, one implementation of the common lottery is to offer each agent a position drawn from the an identical lottery (which does not depend on $\gamma$), then allow agents to accept or reject their offered position. 

In principle, a designer could screen agents on their preference intensity $\gamma$. However, \cref{prop:general-ordinal-preferences-common-lotto} shows that under an analogous convexity condition, the designer also does not benefit from screening agents on $\gamma$: one optimal mechanism is a common lottery.

\begin{proposition}\label{prop:general-ordinal-preferences-common-lotto}
    Suppose $\frac{1}{H(u^{-1}(\cdot; \gamma))}$ is convex on $u(Q;\gamma)$ for all $\gamma \in \Gamma$. For any designer objective function $V$ and agent mass $D$, there exists an optimal mechanism that is a common lottery.
\end{proposition}

The intuition for \cref{prop:general-ordinal-preferences-common-lotto} is as follows. Consider any direct mechanism $a$. If $a$ allocates agents of type $\gamma$ to some position allocation $\bm{s}^\gamma$, \cref{prop:common-lottery-optimal-uneven-grid} implies that the designer may offer agents of type $\gamma$ a common lottery and still attain allocation $\bm{s}^\gamma$ for these agents. Hence the designer may attain the same objective as in $a$ by offering a common lottery to each type $\gamma$. Under this new mechanism, the lottery offered to type $\gamma_1$ may differ from the lottery offered to type $\gamma_2$. However, the designer can then construct \emph{one} common lottery by aggregating the lotteries offered to all $\gamma$: she simply averages over the lotteries offered to each $\gamma$. Because $\gamma$ and $\theta$ are distributed independently, the mass of agents who accept a position of quality $q_k$ is identical in the aggregated common lottery.

Once the designer has used \cref{prop:general-ordinal-preferences-common-lotto} to simplify her search for optimal mechanisms to the space of common lotteries, the same techniques we developed in \cref{subsec:optimal-common-lotteries} characterize the optimal mechanism.\footnote{All that is required is to change the ``demand function'' from $1/F$ to be the function describing the mass of agents who will accept a given position $q_k$, that is, $1/H$.}

\section{Conclusion}\label{sec:future}

We study a designer who assigns agents to a mass of vertically differentiated positions, and has preferences over the mass of agents assigned to each position. Under a condition on the distribution of agents' outside options (the convexity of $1/F$), the designer's optimal mechanism is a common lottery. Common lottery mechanisms are particularly simple to implement, and under the convexity condition, we provide conditions to identify the designer's optimal mechanism. We also provide a partial converse: when the convexity condition does not hold, there exist parameters for which no common lottery is optimal.

Future work could investigate optimal mechanisms when the designer values social welfare instead of utilization, potentially with welfare weights for different types. This may have applications for the allocation of government services such as public housing and food assistance. A designer may also maximize a weighted average of utilization and social welfare. Future work could also investigate when agents vary in horizontal preferences and allow for agents with heterogeneous skill levels (observable or unobservable).

At a conceptual level, many other directions may be of interest to an employer in our setting. An employer may be able to alter the desirability of different positions. For example, instead of offering two heterogeneous positions, she could rotate each employee into both positions. In many settings, she could also vary the wages of each position. Can we quantify the benefits of this flexibility for an employer? How should such an employer optimally use this flexibility?

\newpage
\addcontentsline{toc}{section}{References}
\setlength\bibsep{0pt}
\bibliographystyle{apalike}
\bibliography{refs}

\appendix
\crefalias{section}{appendix}
\crefalias{subsection}{appendix}

\section{Proofs and Examples}

\subsection{Example with 3 agents}
This parametric example finds optimal mechanisms for a simple environment with $3$ agents.

Let $N=3$, the objective be again $V_{\text{fill}}$ and consider type distributions
\[f(\theta_0)=\frac{1}{3},f(\theta_1)=\frac{1}{6}-\epsilon,f(\theta_2)=\frac{1}{2}+\epsilon\]
parameterized by $\epsilon\in(-\frac{1}{2},\frac{1}{6})$. Consider also the two mechanisms in \cref{fig:optmech-example}, where we use short-hand notation $a\vee b:=\max\{a,b\}$ and $a\wedge b:=\min\{a,b\}$.
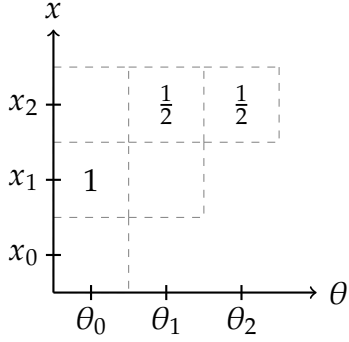
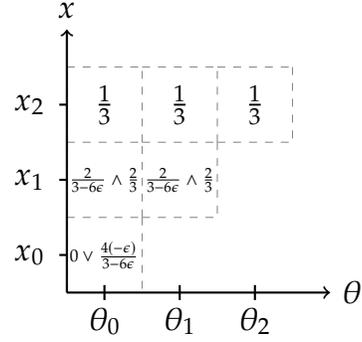
\begin{figure}[htbp]
    \centering
    
    \begin{subfigure}{0.49\textwidth}
        \centering
        \begin{tikzpicture}
        
        \foreach \x in {0,...,3} {
            \foreach \y in {0,...,3} {
            \ifnum\y>\x
                
                \draw[very thin, dashed, gray] (\x,\y-1) -- (\x,\y);
                
                \draw[very thin, dashed, gray] (\x,\y) -- (\x+1,\y);
            \fi
            \ifnum\y=\x
                \ifnum\x<3
                    
                    \draw[very thin, dashed, gray] (\x,\y) -- (\x+1,\y);
                    \draw[very thin, dashed, gray] (\x+1,\y+1) -- (\x+1,\y);
                \fi
            \fi
            }
        }

        \draw[->, thick] (0,0) -- (3.5,0) node[right] {$\theta$};
        \draw[->, thick] (0,0) -- (0,3.5) node[above] {$x$};

        \foreach \t in {0,...,2} {
            \draw[thick] (\t+0.5,-0.1) -- (\t+0.5,0.1) node[below=4pt] {$\theta_{\t}$};
        }

        \foreach \t in {0,...,2} {
            \draw[thick] (-0.1,\t+0.5) -- (0.1,\t+0.5) node[left=4pt] {$x_{\t}$};
        }

            \foreach \x in {1,...,2} {
                \node at (\x+0.5,2.5) {$\tfrac{1}{2}$};
            }

            \foreach \x in {0,...,0} {
                \node at (\x+0.5,1.5) {$1$};
            }
        \end{tikzpicture}
        \caption{Binary-Menu Screening Mechanism}
        \label{fig:optmech-example-binary-appendix} 
    \end{subfigure}
    \begin{subfigure}{0.5\textwidth}
        \centering
        \begin{tikzpicture}
            
            \foreach \x in {0,...,3} {
                \foreach \y in {0,...,3} {
                \ifnum\y>\x
                    
                    \draw[very thin, dashed, gray] (\x,\y-1) -- (\x,\y);
                    
                    \draw[very thin, dashed, gray] (\x,\y) -- (\x+1,\y);
                \fi
                \ifnum\y=\x
                    \ifnum\x<3
                        
                        \draw[very thin, dashed, gray] (\x,\y) -- (\x+1,\y);
                        \draw[very thin, dashed, gray] (\x+1,\y+1) -- (\x+1,\y);
                    \fi
                \fi
                }
            }

            \draw[->, thick] (0,0) -- (3.5,0) node[right] {$\theta$};
            \draw[->, thick] (0,0) -- (0,3.5) node[above] {$x$};

            \foreach \t in {0,...,2} {
                \draw[thick] (\t+0.5,-0.1) -- (\t+0.5,0.1) node[below=4pt] {$\theta_{\t}$};
            }

            \foreach \t in {0,...,2} {
                \draw[thick] (-0.1,\t+0.5) -- (0,\t+0.5) node[left=4pt] {$x_{\t}$};
            }

            \node at (0+0.5,2.5) {$\tfrac{1}{3}$};
            \node at (1+0.5,2.5) {$\tfrac{1}{3}$};
             \node at (2+0.5,2.5) {$\tfrac{1}{3}$};

            \node at (0+0.5,1.5) {{\tiny$\tfrac{2}{3-6\epsilon}\wedge\tfrac{2}{3}$}};
            \node at (1+0.5,1.5) {{\tiny$\tfrac{2}{3-6\epsilon}\wedge\tfrac{2}{3}$}};
            \node at (0+0.5,0.5) {{\tiny $0\vee \tfrac{4(-\epsilon)}{3-6\epsilon}$}};
        \end{tikzpicture}
        \caption{Common Lottery Mechanism}
        \label{fig:optmech-example-commonlotto-appendix} 
    \end{subfigure}
    
    \caption{Optimal Mechanism Example}
    \label{fig:optmech-example}
\end{figure}
We see that when $\epsilon>0$, the binary-menu screening mechanism in the left panel achieves a higher objective value of $\frac{2}{3}$,
while for $\epsilon<0$ the common lottery mechanism from the right panel attains a higher value, of $\frac{2}{3}-\epsilon \frac{4}{9-18\epsilon}>\frac{2}{3}$.\footnote{We derive the value of the binary-menu screening mechanism as $\frac{1}{2}(\frac{1}{6}+\frac{1}{2})+\frac{1}{3}=\frac{2}{3}$, while the value of the common lottery mechanism is similarly found to be $\frac{2}{3}-\epsilon\cdot k,$
where $k=2/3, \text{ if } \epsilon>0$, and $k=\frac{4}{9-18\epsilon}$ otherwise.}
The mechanism designer is indifferent between the two mechanisms when $\epsilon=0$. As it turns out, these mechanisms are also optimal overall in the respective cases. 

\subsection{Proof of Theorem \ref{thm:implementable-by-common-lottery}}\label{proof:implementable-by-common-lottery} 

\begin{proof} 
(\cref{thm:implementable-by-common-lottery}.)
Fix the number of positions/types $N$. Throughout, we zero-index our types and positions.  We show that an arbitrarily chosen feasible solution $a(\cdot;\cdot)$ can be turned into a common lottery as follows. 

Fix a position $x_k$ and recall that $a(x_k;\theta_i)=0$ for types $\theta_i>x_k$.
We will equally distribute the initial mass across all types that accept the offer of $x_k$. 
Formally, we replace the initial lottery with a common lottery $\tilde{c}$ obtained from the degenerate Mean-Preserving Contraction of the initial probabilities $a(x_k;\cdot)$: 
\begin{equation}
    \tilde{c}(x_k) \coloneq \sum_{j=0}^k a(x_k;\theta_j) \frac{f(\theta_j)}{F(\theta_k)} .
\end{equation}
Denote the equivalent DRM of this common lottery by $\tilde{a}$:
$$\tilde{a}(x_k;\theta_i) = 
    \begin{cases}
    \tilde{c}(x_k), & \theta_i \leq x_k \\
    0 , & \theta_i > x_k
\end{cases}.$$

Observe that (i) $\tilde{a}$ allocates the same mass of agents to each position as $a$ does, and so also does not violate position capacity constraints, (ii) $\tilde{a}$ is a common lottery, so it is IC, (iii) $\tilde{a}$ is ex-post IR.

    To prove our conjecture, it remains to show that $\tilde{a}$ does not violate the agent-mass constraints, that is, each type's probability of allocation under allocation rule $\tilde{a}$, $\tilde{P}(\theta_i)$, is at most $1$. 
    Since $\tilde{P}(\theta_i)$ is non-increasing  by construction\footnote{Note this is consistent with our necessary (MON) condition for IC allocations.}, it suffices to show that $\tilde{P}(\theta_0)\leq 1$. 
    We show that our condition and the IC constraints imply this.
    
    Formally, we show something stronger: 
    \begin{equation}
        \tilde{P}(\theta_0) = \sum_{k=0}^{N-1} \tilde{c}(x_k) \leq \sum_{k=0}^{N-1} a(x_k;\theta_0) \leq 1,
    \end{equation}
    where the final inequality follows because the given allocation $a$ is feasible. 
    So it suffices to show the first inequality above, that 
    \begin{equation}\label{eqn:stretching-sufficient-condition}
         \sum_{k=0}^{N-1} a(x_k;\theta_0)  - \sum_{k=0}^{N-1} \sum_{j=0}^k a(x_k;\theta_j)\frac{f(\theta_j)}{F(\theta_k)} \geq 0.
    \end{equation}
    
    To show this, consider our IC constraints for allocation $a$
    \begin{align*}
        \sum_{k=i}^{N-1} h (k-i) (a(x_k;\theta_i)-a(x_k;\theta_{j})) & \geq 0\quad (\text{IC}_{i,j})
    \end{align*}
    with the following coefficients, which also correspond to our Lagrange multipliers, on the relevant domains:
    \begin{align*}
        \lambda_{i,i+1} &\coloneq \frac{f(\theta_{i+1})}{F(\theta_{i+1})} & i \in \{0, 1, \dots, N-2\}\\
        \lambda_{i,j} &\coloneq f(\theta_j)\left(\frac{1}{F(\theta_{i-1})}-\frac{2}{F(\theta_i)}+\frac{1}{F(\theta_{i+1})}\right) & \; \; i\in\{1,2, \dots, N-1\}, j < i \\
        \lambda_{i,j} & \coloneq 0 & \text{otherwise}
    \end{align*}
    The first set are the multipliers on the local upward constraints; the second set are the multipliers on the global downward constraints. For simplicity of notation, we will refer to $\text{IC}_{i,j}$ as the incentive compatibility constraint for type $\theta_i$ not to mimick $\theta_j$, \textit{all taken to the left-hand side and scaled by }$\cdot (N-1)$.\footnote{For example with types $\{0,1/2,1\}$ we will have
\[\text{IC}_{0,1}\coloneq(N-1)\left(\frac{1}{2}(a(x_{1};\theta_0)-a(x_{1};\theta_{1}))+(a(x_2;\theta_0)-a(x_2;\theta_1))\right)|_{N=3}=a({x_{1}};\theta_0)-a(x_{1};\theta_{1})+2(a(x_2;\theta_0)-a(x_2;\theta_1)).\]}
    
    To complete our proof, we show that the left-hand side of inequality \eqref{eqn:stretching-sufficient-condition} is equal to
    \begin{equation}\label{eqn:weighted-IC-constraint-sum}
        \sum_{i=0}^N \sum_{j=0}^N \lambda_{i,j} \cdot (\text{IC}_{i,j}) = \sum_{i=0}^{N-2} \lambda_{i,i+1}  \cdot (\text{IC}_{i,i+1}) + \sum_{i=1}^{N-1} \sum_{j=0}^{i-1} \lambda_{i,j} \cdot (\text{IC}_{i,j})
    \end{equation}
    which is non-negative if each coefficient $\lambda_{i,j}$ we defined is non-negative; $\lambda_{i,i+1}$ is always non-negative, and $\lambda_{i,j}$ is non-negative exactly when our condition holds.\footnote{We show when this condition is necessary for the optimality of a common lottery in  \cref{thm:partial-converse}, our partial converse.}
    We may plug in our IC constraints above, then it remains to check that the coefficients on each term $a(x_k;\theta_i)$ are identical in the LHS of inequality \eqref{eqn:stretching-sufficient-condition} and \cref{eqn:weighted-IC-constraint-sum}. First we consider $i > 0$. Let $\mu(x_k;\theta_i)$ denote the coefficient on $a(x_k;\theta_i)$ in \cref{eqn:weighted-IC-constraint-sum}. Then we have, grouping our terms,

    \begin{align*}
        \mu(x_k;\theta_i) =(k-i)\lambda_{i,i+1} - (k-i+1)\lambda_{i-1,i}+\sum_{j=0}^{i-1}(k-i)\lambda_{i,j} - \sum_{j=i+1}^k (k-j) \lambda_{j,i}
    \end{align*}
    Note that this is the exact expression to show that each constraint in the dual problem is satisfied, if we take a duality approach.

    Plugging in our expressions for the multipliers, the third term immediately simplifies, because the dependence of $\lambda_{i,j}$ on $j$ enters only through $f(\theta_j)$. The fourth term has a nice telescoping form, and simplifies to $-f(\theta_i)\left(\frac{k-i-1}{F(\theta_i)}-\frac{k-i}{F(\theta_{i+1})}+\frac{1}{F(\theta_k)}\right)$. Then plugging in all our multipliers, the entire expression becomes 
    \begin{align}
        \mu(x_k;\theta_i) &\coloneq(k-i)\frac{f(\theta_{i+1})}{F(\theta_{i+1})} - (k-i+1)\frac{f(\theta_{i})}{F(\theta_{i})}+(k-i) F(\theta_{i-1})\left(\frac{1}{F(\theta_{i-1})}-\frac{2}{F(\theta_i)}+\frac{1}{F(\theta_{i+1})}\right) \nonumber \\
        & \quad -f(\theta_i)\left(\frac{k-i-1}{F(\theta_i)}-\frac{k-i}{F(\theta_{i+1})}+\frac{1}{F(\theta_k)}\right) \label{eqn:implementable-by-common-lottery-foc-i-positive} \\ \nonumber \\
        &= - \frac{f(\theta_i)}{F(\theta_k)}, \nonumber
    \end{align}

    which is exactly the coefficient of $a(x_k;\theta_i)$ in \cref{eqn:stretching-sufficient-condition} for $i>0$, as desired.

    For $i=0$, we no longer have the IC constraint $\text{IC}_{i-1,i}$ and no constraints $\text{IC}_{i,j}$ for $j<i$. 
    Therefore our coefficient $\mu(x_k;\theta_0)$ is the same expression, but with the second and third terms omitted. 
    This is then
    \begin{align}
        \mu(x_k;\theta_0) & = k\frac{f(\theta_{1})}{F(\theta_{1})} -f(\theta_0)\left(\frac{k-1}{F(\theta_0)}-\frac{k}{F(\theta_{1})}+\frac{1}{F(\theta_k)}\right) \nonumber \\ \nonumber \\
        & = k \left( \frac{f(\theta_0)+f(\theta_1)}{F(\theta_1)} \right) - (k-1)\frac{f(\theta_0)}{F(\theta_0)} - \frac{f(\theta_0)}{F(\theta_k)} \label{eqn:implementable-by-common-lottery-foc-i-zero} \\
        & = 1 - \frac{f(\theta_0)}{F(\theta_k)}, \nonumber
    \end{align}
    which again matches the coefficients in \cref{eqn:stretching-sufficient-condition}. This concludes the proof of the second sentence of the theorem.
    
    The existence of an optimal mechanism follows from the upper semicontinuity of $V$ and the compactness of our choice set $a$. This implies the final sentence of the theorem.
\end{proof}

\subsection{Duality Proof of \cref{thm:implementable-by-common-lottery}}\label{proof:alternative-implementable-by-common-lotto}

In this section, we offer an alternative proof of \cref{thm:implementable-by-common-lottery} which provides additional intuition. We show the optimality of a common lottery is equivalent to a common lottery solving an auxiliary linear program, then provide interpretable shadow prices on this new linear program.

\begin{proof}(\cref{thm:implementable-by-common-lottery}).
    Consider a vector of \emph{target position masses}, $\bm{s}$, where $s_k$ is the desired mass of agents to allocate to position $k$. For a target position masses $\bm{s}$, we seek the direct mechanism $a$ and the total mass of agents $D$ such that $a$ allocates $\bm{s}$ mass of agents to each position when the total mass of agents is $D$. We seek such a pair so that $D$ is minimized, subject to all other IC and feasibility constraints as well. That is, we aim to fill a target mass of positions while minimizing the mass of agents necessary to fill this mass. 
    
    If we can show that a common lottery $a$ is always an optimum for this problem, regardless of the target position masses $\bm{s}$, then a common lottery will be optimal for any $D$ and objective $V$. This is because if the designer has some other direct mechanism $\tilde{a}$ which attains position masses $\bm{s}$, there exists some common lottery mechanism $a$ which obtains the same position masses $\bm{s}$, and hence the same designer objective $V(\bm{s})$.

    So it suffices to show that if $1/F$ is convex on $\Theta$, then for any $\bm{s}\in\mathbb{R}_+^N$, some common lottery $a$ is a solution to the following program:

    $$
    \min_{D\in \mathbb{R}_+, \; a: \Theta \to \Delta(X \cup \{\varnothing\})} \{ D \}, 
    $$
    subject to constraints:
    \begin{align}
        \sum_{k=i}^{N-1} (x_k - \theta_i)\, a(x_k;\theta_i) 
        & \geq \sum_{k=i}^{N-1} (x_k - \theta_i)\, a(x_k;\theta_j),
        \qquad \forall i,j 
        \tag{$\text{IC}_{i,j}$} \\[6pt]
        D \sum_{i=0}^{k} a(x_k;\theta_i) f(\theta_i) 
        & = s_k,
        \qquad \forall k 
        \tag{Position-Target} \\[6pt]
        \sum_{k=0}^{N-1} a(x_k;\theta_i) 
        & \leq 1,
        \qquad \forall i
        \tag{Agent-Feasibility}
    \end{align}

    Note that we have dropped the Ex-Post IR constraint and incorporated it into the Position-Target constraint, as discussed in the original model section.

    We may then make several adjustments to further simplify the program. First, by \cref{prop:MON}, the IC constraints imply that $Pr(\theta_i) = \sum_{k=0}^{N-1} a(x_k;\theta_i)$ is decreasing in $i$. Hence in the Agent-Feasibility constraints, it suffices only to constrain $Pr(\theta_0)\leq 1$, and we may drop the other constraints as redundant. Second, we will consider the relaxed problem in which the Position-Target fills position $k$ with at least a mass $s_k$ of agents; we will show that a common lottery still solves this relaxed problem, and satisfies the more demanding original Position-Target constraint, and hence solves the original problem. Finally, we divide the Position-Target constraint by $D$ on both sides, which yields an equivalent constraint.\footnote{This is done so that this constraint is clearly expressed as a convex function. We also must have $D >0$ in any solution.} Our relaxed problem, which we refer to as $(*)$, is therefore:

    $$
    \min_{D\in \mathbb{R}_{++}, \; a: \Theta \to \Delta(X \cup \{\varnothing\})} \{ D \}, 
    $$
    subject to constraints:
    \begin{align}
        \sum_{k=i}^{N-1} (x_k - \theta_i)\, a(x_k;\theta_i) 
        & \geq \sum_{k=i}^{N-1} (x_k - \theta_i)\, a(x_k;\theta_j),
        \qquad \forall i,j 
        \tag{$\text{IC}_{i,j}$} \\[6pt]
        \sum_{i=0}^{k} a(x_k;\theta_i) f(\theta_i) 
        & \geq \frac{s_k}{D},
        \qquad \forall k 
        \tag{Position-Target} \\[6pt]
        \sum_{k=0}^{N-1} a(x_k;\theta_0)\ 
        & \leq 1
        \qquad
        \tag{Agent-Feasibility}
    \end{align}

    Note that we can now interpret the Position-Target constraint in terms of probability. The probability that an agent is assigned to position $x_k$ must be at least $s_k/D$, in order to fill the target masses of positions. We may then show the desired multipliers for the Lagrangian of $(*)$. The Lagrangian is

    \begin{align*}
        L(D,a) & = D \\
        & + \sum_{k=0}^{N-1} \lambda_{POS_k} \left( \frac{s_k}{D} - \sum_{i=0}^k a(x_k;\theta_i) f(\theta_i) \right) \\
        & + \lambda_{AGE_0} \left( \sum_{k=0}^{N-1} a(x_k;\theta_0) - 1 \right) \\
        & + \sum_{i=0,j\neq i}^{N-1} \lambda_{\text{IC}_{ij}} \sum_{k=i}^{N-1}(x_k-\theta_i)(a(x_k;\theta_j) - a(x_k;\theta_i))
    \end{align*}

    where $\lambda_{POS_k}$ represents the multiplier on the position-target constraint for position $k$, and $\lambda_{AGE_0}$ is the multiplier on the agent feasibility constraint, and $\lambda_{\text{IC}_{ij}}$ is the multiplier on the IC constraint between type $i$ and $j$. We suppress notation and use $\lambda_{ij}$ to denote $\lambda_{\text{IC}_{ij}}$. 

    In a common lottery, the necessary mass of agents $D$ and the expressions for $a$ are fully determined by feasibility constraints. These terms are

    \begin{align*}
        D & = \sum_{k=0}^{N-1} \frac{s_k}{F(x_k)} \\
        a(x_k;\theta_i) & = \begin{cases}
            \frac{s_k}{D \cdot F(x_k)} & \text{ if } \theta_i \leq x_k \\
            0 & \text{ otherwise}
        \end{cases}\\
    \end{align*}

    This allocation immediately satisfies the IC constraints as it is a common lottery. By construction, it also satisfies the Position-Target and Agent-Feasibility constraints. Our conjectured multipliers to validate the optimality of this common lottery are:

    \begin{align*}
        \lambda_{POS_k} & = \frac{D}{F(\theta_k)} & \forall k\\
        \lambda_{AGE_0} & = D &\\
        \lambda_{\text{IC}_{i,i+1}} &= \frac{f(\theta_{i+1})}{F(\theta_{i+1})} \cdot D \cdot (N-1) & i \in \{0, 1, \dots, N-2\}\\
        \lambda_{\text{IC}_{i,j}} &= f(\theta_j)\left(\frac{1}{F(\theta_{i-1})}-\frac{2}{F(\theta_i)}+\frac{1}{F(\theta_{i+1})}\right) \cdot D \cdot (N-1) & \; \; i\in\{1,2, \dots, N-1\}, j < i \\
        \lambda_{\text{IC}_{i,j}} &= 0 & \text{otherwise}
    \end{align*}

    If $1/F$ is convex, then all conjectured Lagrange multipliers are positive. We provide an interpretation of these multipliers as shadow prices following the proof. Fixing any multipliers, $L$ is convex in $D$ and $a(x_k;\theta_i)$. To verify optimality, it therefore remains to check the first-order conditions on $D$ and $a(x_k;\theta_i)$. The first-order condition on $D$ is straightforward:

    \begin{align*}
        \frac{\partial}{\partial D} L(D,a) & = 1 - \frac{1}{D^2}\sum_{k=0}^{N-1} \lambda_{POS_k} \cdot s_k \\
        & = 1 - \frac{1}{D^2}\sum_{k=0}^{N-1} \frac{D}{F(\theta_k)} \cdot s_k\\
        & = 1 - \frac{1}{D}\sum_{k=0}^{N-1} \frac{s_k}{F(\theta_k)}\\
        & = 0
    \end{align*}
    
    Verifying the optimality of the $a$ terms in the Lagrangian first-order condition uses the same conditions as in \cref{proof:implementable-by-common-lottery}. For $i > 0$ the Lagrangian first-order condition is

    \begin{align*}
         \frac{\partial}{\partial  a(x_k;\theta_i)} L(D,a) & = -\lambda_{POS_k} \cdot f(\theta_i) - \lambda_{i,i+1} (x_k-\theta_i) - \sum_{j=0}^{i-1} \lambda_{i,j} (x_k-\theta_i) \\
         & \quad + \lambda_{i-1,i} (x_k-\theta_{i-1}) + \sum_{j=i+1}^{k} \lambda_{j,i} (x_k - \theta_j)
    \end{align*}

    The first two IC terms represent the IC constraints from type $\theta_i$, which increasing $a(x_k;\theta_i)$ loosens by increasing utility for type $\theta_i$. The latter two IC terms represent the IC constraints towards type $\theta_i$, which increasing $a(x_k;\theta_i)$ tightens. By the exact same computations as in \cref{proof:implementable-by-common-lottery}, specifically \cref{eqn:implementable-by-common-lottery-foc-i-positive}, this expression is indeed zero and the first-order condition is satisfied.

    The case for $i = 0$ is almost identical, except there is an agent-feasibility constraint and no upward IC constraint towards $i=0$. Thus the first-order condition is

    \begin{align*}
         \frac{\partial}{\partial  a(x_k;\theta_0)} L(D,a) & = - \lambda_{POS_k} \cdot f(\theta_i) - \lambda_{i,i+1} (x_k-\theta_i) + \lambda_{AGE_0} + \sum_{j=i+1}^{k} \lambda_{j,i} (x_k - \theta_j)
    \end{align*}

    By the same computations as in \cref{proof:implementable-by-common-lottery} again, specifically \cref{eqn:implementable-by-common-lottery-foc-i-zero}, this first-order condition indeed equals zero and the first-order condition is satisfied.
\end{proof}

An important contribution of this proof is to offer intuition for the shadow prices $\lambda_{ij}$ on the incentive constraints. We describe here an intuition for the shadow prices on the IC constraints, which also provides insight into the origin of the $1/F$ convexity constraint.

First, it should be intuitive that the shadow price on the Position-Target constraint for position $x_k$ is proportional to $\frac{1}{F(x_k)}$. This is because to obtain an additional mass $\varepsilon$ allocated to position $x_k$ in a common lottery, we must offer this position to an additional $\frac{\varepsilon}{F(x_k)}$ of agents, since only $F(x_k)$ of these agents will accept the position. Because we have scaled the constraint by $D$, the shadow price is $\frac{D}{F(x_k)}$ 

Next, for $\lambda_{\text{AGE}_0}$, observe that if we reduce the probability that type $\theta_0$ can be allocated by $\varepsilon$, then to maintain a common lottery we must pick some position $x_k$ and reduce $a(x_k;\theta_i)$ by $\varepsilon$ for all $\theta_i \leq x_k$. This reduces the probability that a randomly selected agent is assigned to position $x_k$ by $F(\theta_k)$. Given that $\lambda_{POS_k}  = \frac{D}{F(\theta_k)}$ as described above, we have the shadow price $\lambda_{\text{AGE}_0} = F(\theta_k) \cdot \frac{D}{F(\theta_k)} = D$.

Next, for the IC constraints, our earlier results in \cref{prop:redundant-ICs} imply that it suffices to focus our attention on the local upward and global downward IC constraints. For the shadow price on the local upward constraints, the intuition is that relaxing the IC constraint from $\theta_i$ to $\theta_{i+1}$ allows the designer to slightly increase the lottery allocation $a(x_{i+1};\theta_{i+1})$. \cref{fig:shadow-price-local-upward-perturbation} displays this perturbation, with the $\bm{+}$ representing the additional mass. When the IC constraint is relaxed, this perturbation is permissible as all types below $\theta_i$ have some slack in their IC constraint to $\theta_{i+1}$, and all types above $\theta_{i+1}$ are unaffected by the perturbation. When the IC constraint is relaxed by $\varepsilon$, we can increase $a(x_{i+1};\theta_{i+1})$ by $\varepsilon\cdot (N-1)$.\footnote{This is because the difference in agent utility between $x_{i+1}$ and $x_i$ is $1/(N-1)$.} The perturbation therefore adds an additional probability $f(\theta_{i+1}) \cdot  \varepsilon \cdot (N-1)$ allocated to position $x_{i+1}$, which buys slack on the Position-Target constraint for $x_{i+1}$. As the shadow price for $POS_{i+1}$ is $\frac{D}{F(x_{i+1})}$, the overall shadow price of this local upward constraint is $\frac{f(\theta_{i+1})}{F(x_{i+1})}\cdot D\cdot (N-1)$.

    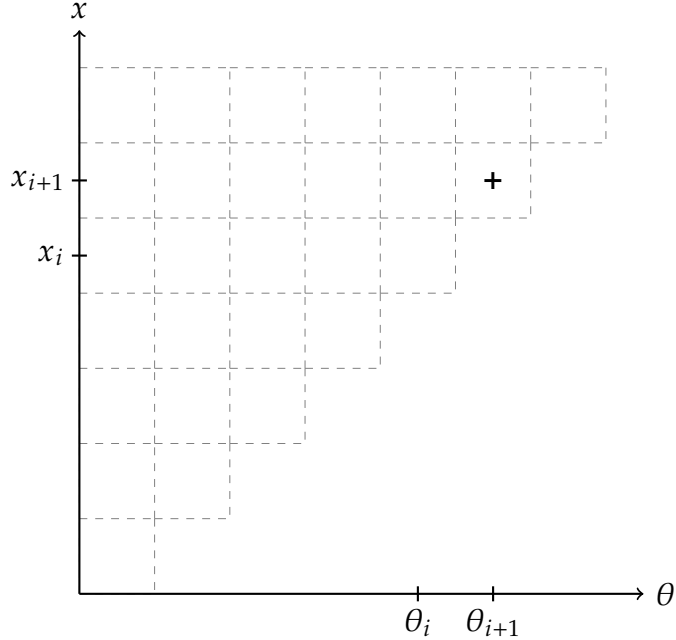
\begin{figure}[h]
    \centering
    \begin{tikzpicture}
      
      \foreach \x in {0,...,7} {
        \foreach \y in {0,...,7} {
          \ifnum\y>\x
            
            \draw[very thin, dashed, gray] (\x,\y-1) -- (\x,\y);
            
            \draw[very thin, dashed, gray] (\x,\y) -- (\x+1,\y);
          \fi
          \ifnum\y=\x
            \ifnum\x<7
                
                \draw[very thin, dashed, gray] (\x,\y) -- (\x+1,\y);
                \draw[very thin, dashed, gray] (\x+1,\y+1) -- (\x+1,\y);
            \fi
          \fi
        }
      }

      \draw[->, thick] (0,0) -- (7.5,0) node[right] {$\theta$};
      \draw[->, thick] (0,0) -- (0,7.5) node[above] {$x$};

      \draw[thick] (4.5,-0.1) -- (4.5,0.1) node[below=4pt] {$\theta_i$};
      
      \draw[thick] (5.5,-0.1) -- (5.5,0.1) node[below=4pt] {$\theta_{i+1}$};

      \draw[thick] (-0.1,4.5) -- (0.1,4.5) node[left=4pt] {$x_i$};
      
      \draw[thick] (-0.1,5.5) -- (0.1,5.5) node[left=4pt] {$x_{i+1}$};

    \node at (5.5,5.5) {$\bm{+}$};
     
    \end{tikzpicture}
    \caption{Perturbation for Shadow Price $\lambda_{\text{IC}_{i,i+1}}$}
    \label{fig:shadow-price-local-upward-perturbation}
    \end{figure}

The intuition for the shadow price on downward IC constraints follows from the convexity of $1/F$, and manipulation of mean-preserving spreads and mean-preserving contractions. If $1/F$ is convex, then if $\bm{s}$ was replaced by a mean-preserving contraction of $\bm{s}$, the designer's objective would improve. That is, using a smaller mass of available agents, the designer can fill the same quantity of positions, but with a mean-preserving contraction of the distribution of which positions are filled.

To maintain the same mass of positions filled while maintaining the original position targets, she could first perform the mean-preserving contraction described above, then perform a mean-preserving spread of the lottery allocated to type $\theta_j$. In general, this violates the IC constraints. However, when we relax the IC constraint $\text{IC}_{i,j}$ for some type $\theta_i$, such a mean-preserving spread becomes possible. \cref{fig:perturbation-shadow-price-downward-IC} displays this perturbation.

    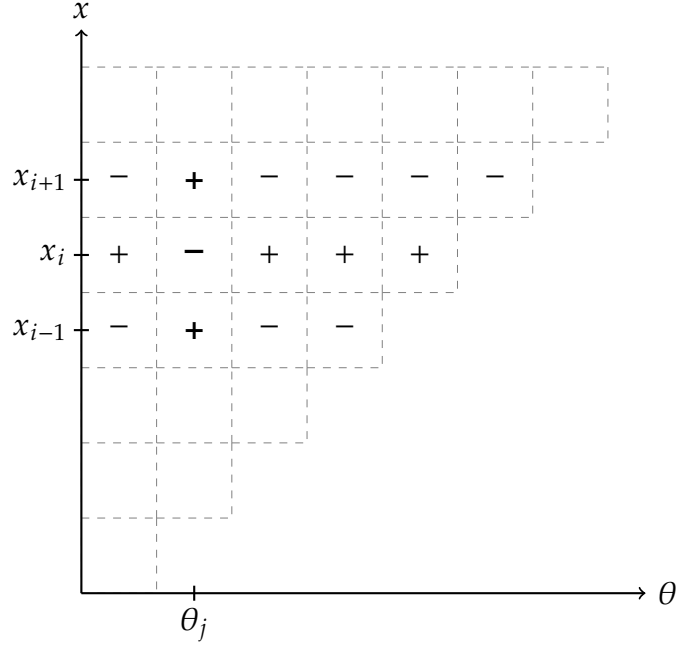
\begin{figure}[h]
    \centering
    \begin{tikzpicture}
      
      \foreach \x in {0,...,7} {
        \foreach \y in {0,...,7} {
          \ifnum\y>\x
            
            \draw[very thin, dashed, gray] (\x,\y-1) -- (\x,\y);
            
            \draw[very thin, dashed, gray] (\x,\y) -- (\x+1,\y);
          \fi
          \ifnum\y=\x
            \ifnum\x<7
                
                \draw[very thin, dashed, gray] (\x,\y) -- (\x+1,\y);
                \draw[very thin, dashed, gray] (\x+1,\y+1) -- (\x+1,\y);
            \fi
          \fi
        }
      }

      \draw[->, thick] (0,0) -- (7.5,0) node[right] {$\theta$};
      \draw[->, thick] (0,0) -- (0,7.5) node[above] {$x$};

      \draw[thick] (1.5,-0.1) -- (1.5,0.1) node[below=4pt] {$\theta_j$};

      \draw[thick] (-0.1,4.5) -- (0.1,4.5) node[left=4pt] {$x_i$};
        
      \draw[thick] (-0.1,5.5) -- (0.1,5.5) node[left=4pt] {$x_{i+1}$};
      
      \draw[thick] (-0.1,3.5) -- (0.1,3.5) node[left=4pt] {$x_{i-1}$};

      \foreach \t in {0,...,5} {
        \ifnum\t=1
          \node at (1.5,5.5) {$\bm{+}$};
        \else
          \node at ({\t + 0.5},5.5) {$-$};
        \fi
      }
      
      \foreach \t in {0,...,4} {
        \ifnum\t=1
          \node at (1.5,4.5) {$\bm{-}$};
        \else
          \node at ({\t + 0.5},4.5) {$+$};
        \fi
      }
      
      \foreach \t in {0,...,3} {
        \ifnum\t=1
          \node at (1.5,3.5) {$\bm{+}$};
        \else
          \node at ({\t + 0.5},3.5) {$-$};
        \fi
      }
    \end{tikzpicture}
    \caption{Perturbation for Shadow Price $\lambda_{\text{IC}_{i,j}}$}
    \label{fig:perturbation-shadow-price-downward-IC}
    \end{figure} 

Suppose the IC constraint $\text{IC}_{i,j}$ is relaxed, so that this mean-preserving spread can move probability $\varepsilon$ from $a(x_{i};\theta_j)$ to $a(x_{i+1};\theta_j)$, and move probability $\varepsilon$ from $a(x_{i};\theta_j)$ to $a(x_{i-1};\theta_j)$. These correspond to a mass $f(\theta_j) \cdot \varepsilon$ moved from position target $s_{i}$ to $s_{i+1}$, and the same mass moved from position target $s_{i}$ to $s_{i-1}$. This means to attain her position targets, the designer needs the remaining agents to fill fewer $x_{i+1}$ and $x_{i-1}$ positions, but to fill more $x_{i}$ positions. The former effect means the designer now requires a mass $ \varepsilon f(\theta_j) \left(\frac{1}{F(x_{i+1})} + \frac{1}{F(x_{i-1})} \right)$ less of agents; the latter effect means she needs an additional $\varepsilon f(\theta_j) \left(\frac{2}{F(x_{i})} \right)$ mass of agents. The sum of these effects is proportional to the shadow price and, when $1/F$ is convex, implies a smaller agent mass is needed to achieve the target position masses $s$.

\subsection{Proof of \cref{thm:partial-converse}}

\begin{proof}
    (\cref{thm:partial-converse}.)
    We will show that there is an improvement when $V$ is strictly increasing, for a particular range of $D$. 
    This uses a variational argument: we begin with the optimal common lottery, then show that if the condition is not satisfied, we can find a perturbation of the common lottery that improves the objective while still satisfying all constraints.
    
    If the condition is not satisfied, then there exists some $k$ (with $k > 0$ and $k < N-1$) such that $\frac{1}{F(\theta_{k})}-\frac{1}{F(\theta_{k-1})} > \frac{1}{F(\theta_{k+1})}-\frac{1}{F(\theta_k)}$. 
    Hence $\frac{1}{F(\theta_{k-1})} - \frac{2}{F(\theta_{k})} + \frac{1}{F(\theta_{k+1})} < 0$. 
    (Incidentally, one can notice that this is one of the multipliers on an inequality in our proof of the direct result; that is why the direct proof does not go through.)

    We then consider the \emph{position} $x_k$, and some type $i < k$. 
    There exists an intermediate interval of $D$ such that the optimal common lottery $a$ satisfies $a(x_{k-1};\theta_{i}), a(x_{k};\theta_{i}), a(x_{k+1};\theta_{i}) > 0$, but $a$ does not fill all positions.\footnote{We will eventually formalize that there is an open interval of such D. Such $D$ require that the following do not hold: 1) zero mass of agents gets allocated at the optimum to the kth position and 2) non-satiation in positions (imbalance of the market towards one side: scarcity of agents).}
    
    Our perturbation will be to decrease $a(x_{k-1};\theta_{i})$, decrease $a(x_{k+1};\theta_{i})$, and increase $a(x_{k};\theta_{i})$, then perform compensating shifts in other allocation cells, in a way that still respects all IC constraints, agent mass constraints, and supply constraints. 
    We will find that $\frac{1}{F(\theta_{k-1})} - \frac{2}{F(\theta_{k})} + \frac{1}{F(\theta_{k+1})} < 0$ exactly implies that this perturbation improves the objective. 
    Hence we will conclude that the optimal common lottery is not optimal for this objective, and hence no common lottery is optimal.

    Our perturbation will be as follows. 
 \cref{fig:perturbation-partial-converse} displays which probabilities increase and decrease in our allocation. 
    First we begin with an intermediate perturbation, with a new allocation which we denote $\tilde{a}$. 
    For some sufficiently small $\varepsilon > 0$ we may choose:
    \begin{align*}
        \tilde{a}(x_k;\theta_i) &= a(x_k;\theta_i) + 2\varepsilon - 2\varepsilon\frac{f(\theta_i)}{F(\theta_k)} & \\
        \tilde{a}(x_{k-1};\theta_i) &= a(x_{k-1};\theta_i)-\varepsilon + \varepsilon \frac{f(\theta_i)}{F(\theta_{k-1})}& \\
        \tilde{a}(x_{k+1};\theta_i) &= a(x_{k+1};\theta_i) - \varepsilon + \varepsilon\frac{f(\theta_i)}{F(\theta_{k+1})} & \\
        \tilde{a}(x_{k};\theta_j) & = a(x_k;\theta_j) - 2\varepsilon\frac{f(\theta_i)}{F(\theta_k)} & \text{for }  j \leq k, j \neq i \\
        \tilde{a}(x_{k-1};\theta_j) &= a(x_{k-1};\theta_j) +\varepsilon \frac{f(\theta_i)}{F(\theta_{k-1})} & \text{for } j \leq k-1, j \neq i \\
        \tilde{a}(x_{k+1};\theta_j) &= a(x_{k+1};\theta_j) + \varepsilon\frac{f(\theta_i)}{F(\theta_{k+1})} & \text{for } j \leq k+1, j \neq i\\
        \tilde{a}(x_{k'};\theta_{j}) & = a(x_{k'};\theta_j)  & \text{for } k' \neq k, k-1, k+1, j \leq k'
    \end{align*}

    The intuition here is that we compress mass onto the cell $(x_k;\theta_i)$ by taking from the cells $(x_{k-1};\theta_i)$ and $(x_{k+1};\theta_i)$. 
    Without any further adjustments, this would violate the capacity constraint of position $x_{k}$ and underfill the positions $x_{k-1}$ and $x_{k+1}$. 
    So we redistribute the transferred mass across all other types allocated to these positions in the common lottery, using the correction terms that are the last terms in each expression. 
    Summing across all allocation terms, this allocation exactly preserves the amount of positions filled. 
    Indeed, it fills each position with the exact same mass of agents as the previous allocation. 
    (In the final step of the proof, we will make a final perturbation to strictly improve the amount of positions filled.)

    We then verify that allocation $\tilde{a}$ satisfies all constraints in our direct revelation mechanism. 
    In particular, we must verify that it satisfies the position capacity constraints, the agent mass constraints, and the IC constraints.

    Because $\tilde{a}$ fills each position with the same mass of agents as allocation $a$ did, all position capacity constraints are satisfied. 
    
    Next we turn to the agent mass constraints. 
    For types greater than $\theta_{k+1}$, their allocated mass does not change. 
    For type $\theta_{k+1}$, his mass increases, but type $\theta_{k+1}$ was not fully allocated under $a$, so there exists $\varepsilon$ sufficiently small that his mass constraint is not violated. 
    For all types less than $\theta_{k}$, their allocated mass increases by $\varepsilon \cdot f(\theta_i)\cdot \left(\frac{1}{F(\theta_{k-1})}-\frac{2}{F(\theta_{k})}+\frac{1}{F(\theta_{k+1})} \right)$. 
    Under our given condition, this is strictly negative, and hence these agents now have slack in their agent mass constraints. 
    For $\theta_k$, the mass allocated to him increases by $\varepsilon \cdot f(\theta_i) \cdot (\frac{1}{F(\theta_{k+1})}-\frac{2}{F(\theta_{k})}) < 0$, and hence his mass constraint is also slack.

    Finally we turn to the IC constraints. 
    Intuitively, the keys here are that 
    (i) all changes are localized to the positions $x_{k-1},x_k,x_{k+1}$, and 
    (ii) for type $\theta_i$, $\tilde{a}(\cdot;\theta_i)$ involves a mean-preserving contraction of $\tilde{a}(\cdot;\theta_j)$, for types $\theta_j < k$ , plus correction terms, 
    (iii) the correction terms are added equally to all types who value these positions. 
    (i) will enable us to satisfy the IC constraints for types greater than $\theta_k$, and (ii) and (iii) will preserve the IC constraints for lower types. 
    
    First note that due to (iii), the allocation for all types other than $\theta_i$ is still a common lottery: types $\theta_j$ and $\theta_{j'}$ are allocated the same lottery over positions that they each care about. 
    Hence IC constraints are satisfied for all $\text{IC}_{j,j'}$ where $j,j'\neq i$. 
    It then remains to check IC constraints to and from $\theta_i$. 
    Type $\theta_i$ has an allocation that is the same common lottery as all other types, with an additional mean-preserving contraction on the positions $x_{k-1},x_k,x_{k+1}$. 
    Because this is a mean-preserving contraction and utilities are linear in position qualities, type $\theta_i$ obtains the same utility as he would from the common lottery offered to any other type $\theta_j$. 
    Hence all IC constraints $\text{IC}_{i,j}$ are satisfied. 
    For any $\text{IC}_{j,i}$ where $j \geq k+1$, all changes made in our perturbation leave the utility of each lottery unchanged, so these IC constraints are all satisfied. 
    For any $\text{IC}_{j,i}$ where $j \leq k-1$, by the same mean-preserving contraction argument, type $j$ obtains the same utility from $\tilde{a}(\cdot;\theta_i)$ as from $\tilde{a}(\cdot;\theta_j)$, and hence these IC constraints are satisfied. 
    The last remaining IC constraint to check is then $\text{IC}_{k,i}$. 
    Then note that for type $\theta_k$, the only payoff-relevant difference between $\tilde{a}(\cdot;\theta_i)$ and $\tilde{a}(\cdot;\theta_k)$ is that $\tilde{a}(x_{k+1};\theta_i) = \tilde{a}(x_{k+1};\theta_k)-\varepsilon$, and hence the former lottery provides strictly lower utility to type $\theta_k$, therefore $\text{IC}_{k,i}$ is satisfied.

    We conclude that $\tilde{a}$ defines an allocation which satisfies all constraints, and obtains the same objective as $a$. 
    We make one final perturbation to construct allocation $\hat{a}$ which satisfies all constraints while strictly increasing the objective. 
    Recall that we chose $D$ so that the optimal common lottery does not fill all positions, and $a(x_{k-1};\theta_{i}) > 0$. 
    Then by the construction of the optimal common lottery there exists some $L \geq 0$ such that allocation $a$ allocates each type $\theta_{j}$ with $j \leq L$ with probability $1$. 
    Furthermore, because $a(x_{k-1};\theta_{i}) > 0$, we have that $L \leq k-1$. 
    We may then choose some position $k' \leq L$ such that position $k'$ is not fully filled.

    For notational convenience, define $\varepsilon' \coloneq -\varepsilon\cdot f(\theta_i)\cdot (\frac{1}{F(\theta_{k-1})}-\frac{2}{F(\theta_{k})}+\frac{1}{F(\theta_{k+1})}) > 0$. 
    From our analysis of the agent mass constraints in $\tilde{a}$ above, we have that $\tilde{a}$ allocates all agents $\theta_j < \theta_k$ with probability at most $1 - \varepsilon'$. 
    In particular, all agents with $\theta_j \leq \theta_{k'}$ are allocated with probability at most $1 - \varepsilon'$. 
    Finally, choosing some sufficiently small $\delta > 0$ with $\delta \leq \varepsilon'$, this slack in the agent mass constraints allows us to construct the new allocation $\hat{a}$ as:
    \begin{equation*}
        \hat{a}(x_k;\theta_j) = 
        \begin{cases}
            \tilde{a}(x_k;\theta_j) + \delta & \text{if } k= k',j\leq k' \\
            \tilde{a}(x_k;\theta_j) & \text{otherwise}
        \end{cases}
    \end{equation*}

    Note that $\hat{a}$ obtains a strictly higher objective than $\tilde{a}$, hence obtains a strictly higher objective than $a$. 
    By our discussion of the agent mass constraints above, all agent mass constraints are still satisfied. 
    For the position capacity constraints, the only position whose fill is changed is $k'$. 
    Since position $k'$ was not fully filled in the allocation $a$, it was not fully filled in allocation $\tilde{a}$, and hence we may choose $\delta$ sufficiently small that its capacity constraint is still satisfied. 
    Finally, we consider IC constraints. 
    Intuitively, this is adding a common term, so IC constraints will be satisfied. 
    Formally, the altered terms do not appear in the IC constraints from any type that is above $\theta_{k'}$ (position $x_{k'}$ lies below their outside options). 
    For types $\theta_j \leq \theta_{k'}$, the IC constraints between these types are still satisfied, because for we have added the same term to each of their allocations; the IC constraints from these types to types above $\theta_{k'}$ are still satisfied, because we have improved the allocation for all types $\theta_j \leq \theta_{k'}$ and left all other allocations unchanged.

    We conclude that allocation $\hat{a}$ satisfies all constraints in our problem, and strictly improves on the objective from $a$. Hence the optimal common lottery mechanism is not optimal in the set of all mechanisms. 
\end{proof}

\subsection{Proof of \cref{prop:optimal-common-lottery}}

\begin{proof}
    For any position mass vector $\bm{s}$ there is a unique common lottery $a$ which yields positions masses $\bm{s}$; the terms of this common lottery are given by 
    $$
    a(x_k;\theta_i) = \begin{cases}
            \frac{s_k}{D \cdot F(x_k)} & \text{ if } \theta_i \leq x_k \\
            0 & \text{ otherwise.}
        \end{cases}
    $$
    
     Hence we may work directly with $\bm{s}$ instead of $a$. The designer's objective is $V(\bm{s})$, so it remains to optimize over the space of $\bm{s}$ that correspond to \textit{feasible} common lotteries. Because $a$ is a common lottery, the IC constraints and ex-post IR constraints are satisfied for any choice of $\bm{s}$. The position capacity constraints are satisfied exactly when $s_k \leq g(x_k)$ for all $k$. The agent mass constraints are satisfied exactly when $1 \geq \sum_{k=0}^{N} a(x_k;\theta_0) = \sum_{k=0}^N \frac{s_k}{D \cdot F(x_k)}$, that is $\sum_{k=0}^N \frac{1}{F(x_k)} s_k \leq D$. These are exactly the objective and two constraints specified in \cref{prop:optimal-common-lottery}.
\end{proof}

\subsection{Proof of \cref{prop:capped-rsd}}

\begin{proof} (\cref{prop:capped-rsd}.)
    In a feasible common lottery, we have 
    \begin{equation*}
        a(x_k; \theta_i) = \begin{cases}
        \frac{1}{D}\frac{s_k}{F(\theta_k)} & \text{ if } \theta_i \leq x_k, \\
        0 & \text{ if } \theta_i > x_k.
    \end{cases}
    \end{equation*}
    
    Under $\text{CRP}(\mathbf{s})$, an agent $\theta_i$'s probability of receiving a position $x_k$ depends on his priority number. 
    The mass of $s_{N-1}$ agents with the highest priority numbers may choose the best position $x_{N-1}$, and they are all willing to accept this position. In notation, $t_{x_{N-1}}^*(1) = s_{N-1}$.
    The probability that type $\theta_i$ is among this mass of agents is the probability that his lottery number is less than $t_{x_{N-1}}^*(1) = s_{N-1}$. This probability is $\frac{s_{N-1}}{D} = \frac{1}{D}\frac{s_{N-1}}{F(x_{N-1})}$.
        
    Once $x_{N-1}$ is exhausted, the next group of agents have a favorite remaining position of $x_{N-2}$. 
    Among this group, only a fraction $F(x_{N-2})$ of them with outside options below $x_{N-2}$ are willing to accept the position. 
    So, under $\text{CRP}(\mathbf{s})$, it takes $\frac{s_{N-2}}{F(x_{N-2})}$ agents to exhaust this position. That is, $t_{x_{N-2}}^*(2) = t^*(1) + \frac{s_{N-2}}{F(x_{N-2})}$.
    The probability that type $\theta_i$ is among this second group is $\frac{1}{D}\frac{s_{N-2}}{F(x_{N-2})}$, and he accepts this position only if $\theta_i \leq x_{N-2}$. 

    This yields exactly the allocation probabilities specified by $a$. Repeating this logic for all positions completes the proof.
\end{proof}

\subsection{Proof of \cref{prop:MON}}

\begin{proof} (\cref{prop:MON}.)
    Fix some $\theta_i < \theta_j$. 
    $\text{IC}_{i,j}$ implies that 
    \begin{align*}
        \sum_{k=i}^{N-1} (x_k - \theta_i) a(x_k;\theta_i) & \geq \sum_{k=i}^{N-1} (x_k - \theta_i) a(x_k;\theta_j), \\ \Rightarrow \quad 
        \sum_{k=i}^{N-1} x_k[a(x_k;\theta_i) - a(x_k;\theta_j)] & \geq \sum_{k=i}^{N-1} \theta_i[a(x_k;\theta_i) - a(x_k;\theta_j)], \\ \Rightarrow \quad 
        \sum_{k=i}^{N-1} x_k[a(x_k;\theta_i) - a(x_k;\theta_j)] & \geq \theta_i (P(\theta_i) - P(\theta_j)). 
    \end{align*}
    $\text{IC}_{j,i}$ implies that 
    \begin{align*}
        \sum_{k=j}^{N-1} (x_k - \theta_j) a(x_k;\theta_j) & \geq \sum_{k=j}^{N-1} (x_k - \theta_j) a(x_k;\theta_i), \\ \Rightarrow \quad 
        \sum_{k=j}^{N-1} x_k[a(x_k;\theta_j) - a(x_k;\theta_i)] & \geq \sum_{k=j}^{N-1} \theta_j[a(x_k;\theta_j) - a(x_k;\theta_i)], \\ \Rightarrow \quad 
        \sum_{k=j}^{N-1} x_k[a(x_k;\theta_j) - a(x_k;\theta_i)] & \geq \theta_j (P(\theta_j) - P(\theta_i)) + \theta_j \sum_{k=i}^{j-1} a(x_k;\theta_i).
    \end{align*}
    Adding these two inequalities, we have 
    \begin{align*}
        \sum_{k=i}^{j-1} x_k[a(x_k;\theta_i) - a(x_k;\theta_j)] & \geq (\theta_j - \theta_i) (P(\theta_j) - P(\theta_i)) + \theta_j \sum_{k=i}^{j-1} a(x_k;\theta_i), \\ \Rightarrow \quad 
        \sum_{k=i}^{j-1} (x_k - \theta_j) a(x_k;\theta_i) & \geq (\theta_j - \theta_i) (P(\theta_j) - P(\theta_i)).
    \end{align*}
    Note that in deriving these inequalities, we use $a(x_k;\theta_j) = 0$ for $x_k < \theta_j$. 
    Since the left-hand side of the final inequality is non-positive, we have that $P(\theta_j) \leq P(\theta_i)$ as desired. 
\end{proof}

\subsection{Proofs of \cref{prop:redundant-ICs} and \cref{prop:binding-ICs}}  The following proof uses local upward ICs together with the allocation monotonicity we derived in \cref{prop:MON} to characterize redundant IC constraints.

\begin{proof} (\cref{prop:redundant-ICs}.) The first claim in the result follows because $\theta_{N-1}$ obtains zero utility from all assignments. For the second claim, we show that for any $\theta_r < \theta_s < \theta_t$, the incentive compatibility constraints $\text{IC}_{r,s}$ and $\text{IC}_{s,t}$, together with the monotonicity constraint, imply $\text{IC}_{r,t}$. 

    Writing $\text{IC}_{r,s}$, we have 
    \begin{align*}
        \sum_{k \geq r} (x_k - \theta_r) a(x_k;\theta_r) & \geq \sum_{k \geq r} (x_k - \theta_r) a(x_k;\theta_s), \\ \Rightarrow \quad 
        \sum_{k \geq r} x_k a(x_k;\theta_r) - \theta_r P(\theta_r) & \geq \sum_{k \geq s} x_k a(x_k;\theta_s) - \theta_r P(\theta_s).
    \end{align*}
    Similarly, from $\text{IC}_{s,t}$, we have 
    \begin{align*}
        \sum_{k \geq s} x_k a(x_k;\theta_s) - \theta_s P(\theta_s) & \geq \sum_{k \geq t} x_k a(x_k;\theta_t) - \theta_s P(\theta_t), \\ \Rightarrow \quad 
        \sum_{k \geq s} x_k a(x_k;\theta_s) - \theta_r P(\theta_s) & \geq \sum_{k \geq t} x_k a(x_k;\theta_t) - \theta_r P(\theta_t),
    \end{align*}
    where in the second inequality, we use the fact that $P(\theta_s) \geq P(\theta_t)$. 
    Combining this last inequality with $\text{IC}_{r,s}$, we have 
    $$ \sum_{k \geq r} x_k a(x_k;\theta_r) - \theta_r P(\theta_r) \geq \sum_{k \geq t} x_k a(x_k;\theta_t) - \theta_r P(\theta_t), $$
    which is exactly $\text{IC}_{r,t}$, as desired. 
    
\end{proof}
The following provides a proof of which ICs bind under the stated sufficient conditions.
\begin{proof}(\cref{prop:binding-ICs})
    Recall our proof of Theorem \ref{thm:implementable-by-common-lottery}, where we replaced an arbitrarily chosen feasible lottery with a feasible common lottery that attains the same objective. Suppose for the sake of contradiction that in an optimal mechanism at least one of the remaining IC constraints does not bind.
    Then the replacement with corresponding common lottery from proof of Theorem \ref{thm:implementable-by-common-lottery} ensures feasibility while simultaneously ensuring that $\tilde{P}(\theta_0)<1$ strictly. This is because $\tilde{P}(\theta_0)<P(\theta_0)\leq 1$, where the first inequality is strict because all coefficients used in the proof are strictly positive (by strict convexity), and at least one IC inequality is strict. Note that the coefficients on local upward ICs are always strictly positive, so the argument for them uses only the weaker version of convexity. Because positions were not filled at the conjectured optimum, we may choose some position $x_k$ and increase the assignment probability $a(x_k;\theta_j)$ for all $j < k$ by some sufficiently small $\epsilon$. This strictly increases the objective yielding a contradiction. 
\end{proof}

\subsection{Proof of \Cref{common-lotto-uniquely-opt}}
\begin{proof}(Corollary \ref{common-lotto-uniquely-opt}) 
We prove the unique optimality of common lotteries under the conditions that $1/F$ be strictly convex, $V$ be strictly increasing in $\bm{s}$ and that assigning all positions at full capacity is infeasible. These conditions imply \cref{prop:binding-ICs} holds, so at any optimal direct mechanism, all local upward and global downward IC constraints bind. Our proof will proceed iteratively starting from highest positions $x$ and showing that the odds of being allocated $x$ must be equalized across all types which see $x$ as an upgrade.

We start with the highest position $x_{N-1}$. From the binding (global) downward IC constraints $\text{IC}_{N-2,\cdot}$, at any optimal allocation $a(\cdot;\cdot)$,
\[a(x_{N-1};\theta_{N-2})=a(x_{N-1};\theta_j),\;\; \forall_{j\neq N-1}\]
therefore at the optimum
 \[a(x_{N-1};\theta_{k'})\equiv b(x_{N-1}), \;\; \forall_{k'\leq N-1}\]
 independently of type.

 We may conclude the proof by induction. Our inductive hypothesis for $n \geq 1$ is that for all positions with index $k \geq N-n$ and all types $m \leq k$, 
 \[a(x_k;\theta_m)\equiv b(x_k).\]
 By the same argument as above, the binding downward IC constraints $\text{IC}_{k-2,\cdot}$ imply that
 \[a(x_{k-1};\theta_m)\equiv b(x_{k-1}),\]
 for all $m\leq k-1$. We conclude that $a$ is a common lottery.
\end{proof}

\subsection{Equivalence result under $V_{\text{fill}}$}\label{app:V_fill}
We first define operations on allocations used in proving this auxiliary result. Let us consider an arbitrary feasible allocation $a(\cdot;\cdot)$. 
We introduce two natural ways of altering the allocation:\footnote{We will investigate the connection of these operations to dominance in the convex order.}
\begin{enumerate}
    \item \textbf{Allocation upgrade for type $\theta$} constitutes an FOSD shift in the lottery for type $\theta$. 
    In other words, odds of type $\theta$ of getting some worse positions are swapped for odds of getting some better positions --- type $\theta$ gets "upgraded".
    \item \textbf{Equalizing chances for position $x$} constitutes a (degenerate) Mean-Preserving Contraction\footnote{To see  that $\tilde{a}$ is position-wise a Mean-Preserving Contraction of $a$, fix $x$. 
    Using definition of $\tilde{a}(x;\theta)$, we get
\[\mathbb{E}_{\theta}[a(x;\theta)]=\mathbb{E}_{\theta}[\mathbf{1}_{\theta\leq x}a(x;\theta)+\mathbf{1}_{\theta>x}a(x;\theta)]=\mathbb{E}_{\theta}[\mathbf{1}_{\theta\leq x}\tilde{a}(x;\theta)+\mathbf{1}_{\theta>x}\tilde{a}(x;\theta)]=\mathbb{E}_{\theta}[\tilde{a}(x;\theta)].\]
} at position $x$, and is formally defined by
    \[\tilde{a}(x';\theta)\coloneq\begin{cases}
    a(x';\theta), \text{ if } x'\neq x\\
    \sum_{\theta'=0}^x\frac{a(x;\theta')f(\theta')}{F(x)},\text{ if } x'=x, \theta\leq x\\
    0,\text{otherwise}\end{cases}\] 
    Intuitively, this operation equalizes odds of getting position $x$ among types $\theta\leq x$ who deem $x$ acceptable.
\end{enumerate}
The operation equalizing chances for \textit{all} positions will be called \textbf{allocation shrinking}. 
If an allocation upgrade for \textit{any} type $\theta$ precedes allocation shrinking, we will talk of \textbf{upgraded allocation shrinking}. 
Note that the former is trivially nested by the latter by making a measure $0$ upgrade. 
We say that \textbf{allocation shrinking is achievable} if, starting from any feasible allocation $a(\cdot;\cdot)$, one can apply an upgraded allocation shrinking and preserve feasibility.

Let us define a common lottery $\pi\in\Delta(X\cup\{\varnothing\})$ by
\[\pi(x_k)\coloneq\begin{cases}
    \frac{g(x_k)}{D F(x_k)}, \text{ if } k>\underline{k}\\
    \min\{\frac{g(x_k)}{D F(x_k)},1-\sum_{k'=\underline{k}+1}^{n-1}\frac{g(x_{k'})}{D F(x_{k'})}\}, \text{ if } k=\underline{k}\\
    0, \text{otherwise,}
\end{cases}\]
where $\underline{k}\coloneq\min\left(\{k': \sum_{k=k'+1}^{n-1}\frac{g(x_k)}{D F(x_k)}\leq 1\leq \sum_{k=k'}^{n-1}\frac{g(x_k)}{D F(x_k)}\}\cup\{0\}\right)$.

\begin{lemma}\label{lem:shrinking-iff-common_lotto_opt}
    Common lottery $\pi$ is optimal under $V_{\text{fill}}$ $\iff$ allocation shrinking is achievable.
\end{lemma}

\begin{proof} (Lemma \ref{lem:shrinking-iff-common_lotto_opt})
$(\;\Longleftarrow \;:)$ First suppose that allocation shrinking is feasible, that is, for any feasible allocation $a(\cdot;\cdot)$ there exists some upgraded allocation shrinking $\tilde{a}(\cdot;\cdot)$ for which
\[\tilde{P}(\theta_0)\coloneq\sum_{k=0}^{N-1}\tilde{a}(x_k;\theta_0)\leq 1.\]
Now, utilization at $a(\cdot;\cdot)$ is equal to
\[D\sum_{k=0}^{N-1}\tilde{a}(x_k;\theta_0)F(x_k),\]
we see that within the class of common lotteries\footnote{It is important to distinguish the probability of being assigned to any position (including the one valued at a player’s outside option) from the probability of staying with the outside option. Thus, these outcomes constitute lotteries, even when the allocation probability is strictly below one for the lowest type.}, the common lottery $\pi$ 1) pushes maximal masses to fill highest positions \textit{and} 2) the probability that the lowest type gets allocated to a position is maximized (and $\sum_{j=0}^n\pi(j)=1$ whenever feasible), so it maximizes the objective.\footnote{Note we needed to modify the above to account for the case where allocating to lowest type with probability $1$ is infeasible due to binding position capacity constraints.}\\ 

\noindent $(\implies :)$
Now suppose that common lottery $\pi$ is optimal. This means that for every feasible allocation $a(\cdot;\cdot)$, utilization satisfies \[D\sum_{k=0}^{N-1}\sum_{i=0}^{N-1}a(x_k;\theta_i)f(\theta_i)\leq D\sum_{k=0}^{N-1}\pi(x_k)F(x_k)=OptVal.\]

Note however that it might be that $a(\cdot;\cdot)$ allocates some positions $k<\underline{k}$, so we cannot yet simplify to position-by-position comparisons. Let us therefore perform an upgraded allocation shrinking in the following way. Let us first enact a maximal upgrading of types\footnote{This is allocation dependent. The order of picking types to upgrade is irrelevant for the argument. Formally, a maximal upgrading of types is defined via the following algorithm: Terminate if $\{k: D \sum_{i=0}^{k}a(x_k;\theta_i)f(\theta_i)<g(x_k)\}$. Otherwise, initialize by fixing the highest partly vacant position $\tilde{k}\coloneq\max\{k: D \sum_{i=0}^{k}a(x_k;\theta_i)f(\theta_i)<g(x_k)\}$. Step 1) Pick any type $\theta'\in \{\theta_i: a(x_k;\theta_i)>0 \text{ for some }k<\tilde{k}\}$  (if this set is empty, terminate) and upgrade them to $x_{\tilde{k}}$ by swapping odds until we either 1a) run out of odds to swap for $\theta'$ or 1b) position $x_{\tilde{k}}$ is filled at its capacity. If the former 1a) is the case, pick another type arbitrarily and do the same. If the latter condition 1b) holds, update $\tilde{k}\rightarrow \tilde{k}-1$ (if exists, i.e., provided $\tilde{k}-1\geq 0$) and go to step 1).} and shrink the resulting allocation. Note that these operations preserve utilization and fill at capacity maximal number of positions subject to this utilization. This also allows us to compare the two allocations position-by-position, and by optimality of $\pi$ it must be that $\tilde{P}(0)\leq P_{\pi}(0)\leq 1$, hence allocation shrinking is achievable.    
\end{proof}
x
\section{General Designer Objectives}\label{app:general-designer-objectives}

If our convexity condition is satisfied, we may use standard convex optimization techniques (\cite{boyd2004convex}) to determine the optimal mechanism. It suffices to simply find the shadow price on the budget constraint in \cref{prop:optimal-common-lottery}. 

\begin{proposition}\label{prop:optimal-mechanism-general-objective}
    Suppose $1/F$ is convex and $V$ is concave and differentiable. The common lottery $a$ with feasible allocation vector $\bm{s}$ is an optimal mechanism if and only if there exists $\lambda \geq 0$ such that $\lambda \left(\sum_{k=0}^{N-1} \frac{s_k}{F(x_k)} - D\right) = 0$ and for all $k \in \{0, \dots, N-1\}$,

    \begin{enumerate}[(i)]
        \item if $s_k \in (0, g(x_k))$, then $V_k(\bm{s}) = \frac{\lambda}{F(x_k)}$,
        \item if $s_k = 0$, then $V_k(\bm{s}) \leq \frac{\lambda}{F(x_k)}$,
        \item if $s_k = g(x_k)$, then $V_k(\bm{s}) \geq \frac{\lambda}{F(x_k)}$.
    \end{enumerate}
\end{proposition}

By a feasible allocation vector $\bm{s}$, we refer to a vector $\bm{s}$ that satisfies the conditions of  \cref{prop:optimal-common-lottery}: $s_k \in [0, g(x_k)]$ for all $k$ and $\sum_{k=0}^{N-1} \frac{s_k}{F(x_k)} \leq D$.

\begin{proof}(\cref{prop:optimal-mechanism-general-objective}.)
    It suffices to find an optimal $\bm{s}$ for the optimization problem given in \cref{prop:optimal-common-lottery}. Because $V$ is concave and all constraints are linear, the Karush-Kuhn-Tucker conditions are sufficient and necessary to characterize an optimal solution.

    We may write the Lagrangian as
    $$
    \mathcal{L}(\bm{s}, \lambda, \bm{\alpha}, \bm{\beta}) = V(\bm{s}) - \lambda \left(\sum_{k=0}^{N-1} \frac{s_k}{F(x_k)} - D\right) - \sum_{k=0}^{N-1} \alpha_k(s_k - g(x_k)) + \sum_{k=0}^{N-1} \beta_k(s_k)
    $$

    By the complementary slackness condition for $\lambda$, we have $\lambda \left(\sum_{k=0}^{N-1} \frac{s_k}{F(x_k)} - D\right) = 0$. By the first-order conditions on $s_k$,
    $$
    V_k(\bm{s}) - \frac{\lambda}{F(x_k)} - \alpha_k +\beta_k = 0
    $$

    In case (i), $\alpha_k = \beta_k = 0$ so $V_k(\bm{s}) = \frac{\lambda}{F(x_k)}$. In case (ii), $\alpha_k = 0$ and $\beta_k \geq 0$ so $V_k(\bm{s}) \leq \frac{\lambda}{F(x_k)}$. In case (iii), $\alpha_k \geq 0$ and $\beta_k = 0$ so $V_k(\bm{s}) \geq \frac{\lambda}{F(x_k)}$. If (i)-(iii) are satisfied, we may then choose $\bm{\alpha}$ and $\bm{\beta}$ to satisfy the resulting first-order conditions on $s_k$.
\end{proof}

To take one example, if $V$ is linear and additive in $s_k$, with $V(\bm{s}) = \sum_{k=0}^{N-1} \alpha_k s_k$, then the optimal mechanism first orders the positions by their value $\alpha_k / F(x_k)$, and fully fills the highest-value position for the designer, then the next highest-value position for the designer, and so forth, until exhausting the budget constraint $\sum_{k=0}^{N-1} \frac{s_k}{F(x_k)} \leq D$.

As a corollary, we can generalize this approach when the designer has flexible capacity in each position, that is, if $g(x_k) = \infty$ for each position. She still may have decreasing returns to adding agents to each position.

\begin{corollary}\label{cor:optimal-mechanism-flexible-capacity}
    Suppose $1/F$ is convex and $V$ is concave and differentiable. Suppose $g(x_k) = \infty$ for all $k$. A common lottery $a$ with feasible allocation vector $\bm{s}(a)$ is an optimal mechanism if and only if $\lambda \left(\sum_{k=0}^{N-1} \frac{s_k}{F(x_k)} - D\right) = 0$ and for all $k \in \{0, \dots, N-1\}$ and for all $k \in \{0, \dots, N-1\}$,

    \begin{enumerate}[(i)]
        \item if $s_k > 0$, then $V_k(\bm{s}) = \frac{\lambda}{F(x_k)}$,
        \item if $s_k = 0$, then $V_k(\bm{s}) \leq \frac{\lambda}{F(x_k)}$.
    \end{enumerate}
\end{corollary}

Even when $V$ is not concave, \cref{prop:optimal-common-lottery} still simplifies the designer's problem to optimization of vector $\bm{s} \in \mathbb{R}_+^n$ with one budget constraint.

\section{Implementation of Common Lotteries}\label{app:implementation-common-lotteries}

Formally defining the random priority mechanism in the continuum setting requires care; we use the same techniques as \cite{che2010asymptotic} to define the random priority mechanism in a continuum setting. The caps on each position become the available quota of each position. Agents each draw a lottery number uniformly and independently from $[0,1]$.\footnote{As described in a \cite{uhlig1996continuum}, we may appeal to a law of large numbers to draw each type independently and guarantee that the aggregate distribution of agents' lottery numbers is the uniform distribution.} The mechanism then iteratively assigns masses of agents to the available positions, assigning each agent to his most-preferred position among the remaining positions that he is willing to accept. For a position-feasible target vector $\bm{s}$ specifying the masses of agents assigned to each position,\footnote{A vector $\bm{s}$ is position-feasible if $s_k \leq g(x_k)$, for all $k$.} we define the following random-priority mechanism with capping vector $\bm{s}$.

\begin{definition}
    A \emph{Capped Random Priority} mechanism with capping vector $\mathbf{s}$, denoted $\text{CRP}(\mathbf{s})$, operates as follows. Assign each agent a lottery number drawn uniformly and independently from $[0,1]$. Set the initial quota of each position $x_k$ to be $s_k$ and the initial set of available positions to be $\hat{X} = X$. Beginning with $m = 1$, iterate through the following steps.

    At step $m$, for each position $x \in \hat{X}$, determine threshold $t_{x}^*(m)$ such that the remaining quota of position $x$ is equal to the measure of agents whose favorite remaining position is $x$, prefer position $x$ to their outside option, and have lottery number less than $t_{x}^*(m)$. If no such value exists, let $t_{x}^*(m) := 1$. Let $t^*(m) = \min_{x \in \hat{X}} t_{x}^*(m)$. Assign each agent with lottery number less than $t^*(m)$ to his preferred option among the remaining positions and outside option. If all agents have been assigned, stop. Otherwise, remove the assigned agents, reduce the quota of each position $x_k$ by the mass of agents assigned to $x_k$, remove any position with remaining quota of $0$ from $\hat{X}$, and continue to step $m+1$.
\end{definition}

Because the set of positions $X$ is finite, this procedure terminates in a finite number of steps. In our setting with vertically differentiated positions, the capped random priority mechanism takes a particularly simple form. Each agent claims the highest-quality position available to them, or his outside option, when he is called to choose his allocation. In our notation, the first round of the mechanism fills position $x_{N-1}$, the next round fills position $x_{N-2}$, and so forth.

\section{Common Ordinal Preferences}\label{app:common-ordinal-preferences}

\subsection{Proof of \cref{prop:common-lottery-optimal-uneven-grid}}

\begin{proof}(\cref{prop:common-lottery-optimal-uneven-grid}.)
    We work in utility space throughout the proof. We define $x_k = u(q_k)$, and we use $\theta_i = u(q_i)$ to denote an agent's utility from outside option $q_i$. The distribution of agents willing to accept a position $x_k$ is given by $H(u^{-1}(x_k))$, so we may define $F$ as $H(u^{-1}(\cdot))$. This allows us to proceed with the same expression of the direct mechanism as in our main proof.

    As in the alternative proof of our main theorem, it suffices to show that for a given vector of position target allocation targets $\bm{s}$, a common lottery is optimal for the problem of allocating $\bm{s}$ while minimizing the mass of agents initially available. Given that we are working in utility space, we may write the same convex optimization program as in that proof, then apply the same relaxations to obtain the problem (*). For clarity, we copy this problem below:

    $$
    \min_{D\in \mathbb{R}_{++}, \; a: \Theta \to \Delta(X \cup \{\varnothing\})} \{ D \}, 
    $$
    subject to constraints:
    \begin{align}
        \sum_{k=i}^{N-1} (x_k - \theta_i)\, a(x_k;\theta_i) 
        & \geq \sum_{k=i}^{N-1} (x_k - \theta_i)\, a(x_k;\theta_j),
        \qquad \forall i,j 
        \tag{$\text{IC}_{i,j}$} \\[6pt]
        \sum_{i=0}^{k} a(x_k;\theta_i) f(\theta_i) 
        & \geq \frac{s_k}{D},
        \qquad \forall k 
        \tag{Position-Target} \\[6pt]
        \sum_{k=0}^{N-1} a(x_k;\theta_0)\ 
        & \leq 1
        \qquad
        \tag{Agent-Feasibility}
    \end{align}

    The Lagrangian is again identical to the baseline proof. Our conjectured optimal $D$ and $a$ are identical as well, again pinned down by a common lottery. Our conjectured Lagrange multipliers on the constraints are the same with only the IC constraints updated:

    \begin{align*}
        \lambda_{POS_k} & = \frac{D}{F(\theta_k)} & \forall k\\
        \lambda_{AGE_0} & = D &\\
        \lambda_{\text{IC}_{i,i+1}} & = \frac{f(\theta_{i+1})}{F(\theta_{i+1})} \cdot \frac{D}{x_{i+1} - x_i}& i \in \{0, 1, \dots, N-2\} \\
        \lambda_{\text{IC}_{i,j}} & = f(\theta_j) \cdot \left( \frac{x_{i+1} - x_i}{F(\theta_{i-1})} - \frac{x_{i+1} - x_{i-1}}{F(\theta_i)} + \frac{x_i - x_{i-1}}{F(\theta_{i+1})} \right)\cdot \frac{D}{(x_{i+1} - x_i)(x_i - x_{i-1})} & \; i\in\{1,2, \dots, N-1\}, j < i \\
        \lambda_{\text{IC}_{i,j}} & = 0 & \text{otherwise}
    \end{align*}

    Note that in the case of our baseline model where $x_k = \frac{q}{N-1}$ are evenly spaced, the IC multipliers are identical to our baseline IC multipliers. The convexity of $1/F$ guarantees that $\lambda_{IC_{ij}}$ is non-negative. As the Lagrangian is convex, it remains only to check the first-order conditions on $D$ and $a$. The first-order condition on $D$ is identical to our baseline proof, and as $\lambda_{POS_k}$ is identical as well, this condition is satisfied. It remains only to check that the first-order conditions on $a$ are satisfied. For $i>0$ the expression for this first-order condition is
    \begin{align*}
         \frac{\partial}{\partial  a(x_k;\theta_i)} L(D,a) & = -\lambda_{POS_k} \cdot f(\theta_i) - \lambda_{i,i+1} (x_k-\theta_i) - \sum_{j=0}^{i-1} \lambda_{i,j} (x_k-\theta_i) \\
         & \quad + \lambda_{i-1,i} (x_k-\theta_{i-1}) + \sum_{j=i+1}^{k} \lambda_{j,i} (x_k - \theta_j)
    \end{align*}

    The final term $\sum_{j=i+1}^{k} \lambda_{j,i} (x_k - \theta_j)$ again telescopes, and now simplifies to\linebreak $\frac{f(\theta_i)\cdot D}{x_{i+1}-x_i} \left( \frac{x_k - x_{i+1}}{F(\theta_i)} - \frac{x_k - x_i}{F(\theta_{i+1})} + \frac{x_{i+1} - x_i}{F(\theta_k)} \right)$. Our entire expression is then

    \begin{align*}
         \frac{\partial}{\partial  a(x_k;\theta_i)} L(D,a) & = -D \cdot \frac{f(\theta_i)}{F(\theta_k)} - \frac{f(\theta_{i+1})}{F(\theta_{i+1})} \cdot \frac{D(x_k-\theta_i)}{x_{i+1} - x_i} \\
         & \quad - (x_k - \theta_i) F(\theta_{i-1})\left( \frac{x_{i+1} - x_i}{F(\theta_{i-1})} - \frac{x_{i+1} - x_{i-1}}{F(\theta_i)} + \frac{x_i - x_{i-1}}{F(\theta_{i+1})} \right)\cdot \frac{D}{(x_{i+1} - x_i)(x_i - x_{i-1})} \\
         & \quad + \frac{f(\theta_{i})}{F(\theta_{i})} \cdot \frac{D}{x_{i} - x_{i-1}}\cdot (x_k-\theta_{i-1}) \\
         & \quad + \frac{f(\theta_i)\cdot D}{x_{i+1}-x_i} \left( \frac{x_k - x_{i+1}}{F(\theta_i)} - \frac{x_k - x_i}{F(\theta_{i+1})} + \frac{x_{i+1} - x_i}{F(\theta_k)} \right).
    \end{align*}

    Combining terms by their denominator yields:
    \begin{align*}
        \frac{1}{D} \cdot \frac{\partial}{\partial  a(x_k;\theta_i)} L(D,a) & = -\frac{1}{F(x_{i-1})} \left( F(x_{i-1})\frac{x_k - x_i}{x_i - x_{i-1}} \right) \\
        & \quad+ \frac{1}{F(x_i)}\left( F(x_{i-1}) \frac{(x_k - x_i)(x_{i+1}-x_{i-1})}{(x_{i+1} - x_i)(x_i - x_{i-1})} + f(x_i) \frac{x_k - x_{i-1}}{x_i - x_{i-1}} + f(x_i) \frac{x_k - x_{i+1}}{x_{i+1} - x_i} \right) \\ 
        & \quad - \frac{1}{F(x_{i+1})} \left(F(x_{i-1}) \frac{x_k - x_i}{x_{i+1} - x_i} + f(x_i) \frac{x_k - x_i}{x_{i+1} - x_i}  + f(x_{i+1}) \frac{x_k - x_i}{x_{i+1} - x_i}   \right) \\ \\
        & = -\frac{x_k - x_i}{x_i - x_{i-1}} + \frac{x_k - x_i}{x_i - x_{i-1}} + \frac{x_k - x_i}{x_{i+1} - x_i} - \frac{x_k - x_i}{x_{i+1} - x_i} \\ 
        \\
        & = 0.
    \end{align*}

    For $i =0$, the first-order condition is 
        \begin{align*}
         \frac{\partial}{\partial  a(x_k;\theta_0)} L(D,a) & = - \lambda_{POS_k} \cdot f(\theta_i) - \lambda_{i,i+1} (x_k-\theta_i) + \lambda_{AGE_0} + \sum_{j=i+1}^{k} \lambda_{j,i} (x_k - \theta_j)\\
         & = - D\frac{f(\theta_i)}{F(\theta_k)} - \frac{f(\theta_{i+1})}{F(\theta_{i+1})} \cdot \frac{D(x_k-\theta_i)}{x_{i+1} - x_i} \\
         & \quad + D + \frac{f(\theta_i)\cdot D}{x_{i+1}-x_i} \left( \frac{x_k - x_{i+1}}{F(\theta_i)} - \frac{x_k - x_i}{F(\theta_{i+1})} + \frac{x_{i+1} - x_i}{F(\theta_k)} \right) \\
         & = D \left( 1 + \frac{f(x_0)}{F(x_0)}\cdot \frac{x_k - x_{1}}{x_{1} - x_0} - \frac{f(x_0) + f(x_{1})}{F(x_{1})} \cdot \frac{x_k - x_0}{x_{1} - x_0}\right) \\
         & = 0.
    \end{align*}

    We conclude that the given $D$ and $a$ are feasible and optimal given multipliers $\lambda$, and hence are optimal for problem (*).
\end{proof}

\subsection{Direct Mechanisms with Common Ordinal Preferences}\label{app:subsec-direct-mechs-common-ordinal}

We let $h^\Gamma: \Gamma \to [0,1]$ denote the probability mass function of agents' marginal distribution on $\Gamma$. We let $h^Q: Q \to [0,1]$ denote agents' marginal distribution on outside options. Under our independence assumption, these two marginal distributions fully describe the joint distribution. More generally, we let $h: Q \times \Gamma \to [0,1]$ denote the joint distribution of agents' outside options and $\gamma$ parameter. 

In our augmented model, a direct mechanism $a: Q \times \Gamma \to \Delta(Q \cup \{ \varnothing \})$ satisfies the following constraints:

\begin{align*}
    \sum_{k=i}^{N-1} (u(q_k;\gamma) - u(q_i; \gamma))\, a(x_k;q_i,\gamma) 
    & \geq \sum_{k=i}^{N-1} (u(q_k;\gamma) - u(q_i; \gamma))\, a(x_k; q_j, \gamma'), \\
    & \quad \quad \forall \; i,j \in \{0, \dots, N-1\}, \gamma,\gamma' \in \Gamma
    \tag{$\text{IC}_{i,j}$} \\[6pt]
    D \sum_{i=0}^{N-1} \sum_{\gamma \in \Gamma} a(q_k;q_i, \gamma) h(q_i,\gamma)
    & \leq g(q_k),
    \qquad \forall \; k  \in \{0,\dots,N\}
    \tag{Position-Feasibility} \\[6pt]
    \sum_{k=0}^{N-1} a(q_k;q_i, \gamma)\, 
    & \leq 1,
    \qquad \forall \; i \in \{0,\dots,N\}, \gamma \in \Gamma
    \tag{Agent-Feasibility} \\[6pt]
    a(q_k;q_i, \gamma) & = 0,
    \qquad \forall \; k<i, \forall \; \gamma \in \Gamma
    \tag{Ex-Post IR}
\end{align*}

\subsection{Proof of \cref{prop:general-ordinal-preferences-common-lotto}}

\begin{proof} (\cref{prop:general-ordinal-preferences-common-lotto}.)
  We show that a common lottery attains allocation $\bm{s}$ in the relaxed problem when $\gamma$ is observable by the designer. As a common lottery remains incentive compatible when $\gamma$ is unobservable, the same common lottery attains allocation $\bm{s}$ in our original problem.

  Consider any direct mechanism $a: Q \times \Gamma \to \Delta(X \cup \{\emptyset\})$. In this mechanism, agents of type $\gamma$ are allocated to fill some mass $\bm{s}_\gamma$ of positions. By \cref{prop:common-lottery-optimal-uneven-grid}, there exists a mechanism $a_\gamma: \Theta \to \Delta(X \cup \{\emptyset\})$ such that $a_\gamma$ is a common lottery and $a_\gamma$ allocates the agents of type $\gamma$ to fill $\bm{s}_\gamma$.  Let $\hat{a}$ denote the mechanism that assigns all agents of type $\gamma$ using lottery $a_\gamma$, so $\hat{a}(q_k;v,\gamma) = a_\gamma(q_k;v)$. If $\gamma$ is observable, then $\hat{a}$ is feasible, incentive compatible, and attains the original allocation: $\bm{s}(\hat{a}) = \bm{s}(a)$.

  The mechanism $\hat{a}$ is not necessarily incentive-compatible when $\gamma$ is unobservable. However, we can transform it into a true common lottery, which is incentive compatible, by simply taking the average of the lotteries given by $a_\gamma$. The independence of $\gamma$ and an agent's outside option will then imply that the same mass of agents accept their offers as before. Let $v$ denote the quality of an agent's outside option. Then we define $\tilde{a}$ as

  $$
  \tilde{a}(q_k; v, \gamma) = \sum_{\gamma \in \Gamma} a_\gamma(q_k; v) h^\Gamma(\gamma)
  $$
  
  $\tilde{a}$ is a common lottery, and only assigns a position of quality $q_k$ to an agent of type $v$ when $q_k \geq v$, because each $a_\gamma$ satisfies this constraint.\footnote{Suppose $q_k < v$. Then $a_{\gamma}(q_k;v) = 0$ for all $\gamma \in \Gamma$, hence $\tilde{a}(q_k;v,\gamma) = 0$ as well.} It is a common lottery, so it is incentive-compatible: each agent is offered the same lottery, truncated below at the quality of his outside option $v$. Finally, $\tilde{a}$ fills positions with a mass of $\bm{s}(a)$, as it is a weighted average of the $a_\gamma$ mechanisms. Formally,
  \begin{align*}
      s_k(\tilde{a}) & =  D \sum_{\gamma \in \Gamma} \sum_{v =q_0}^{q_k} h^\Gamma(\gamma) h^Q(v) \tilde{a}(q_k;v,\gamma)\\
      & = D \sum_{\gamma \in \Gamma} \sum_{v =q_0}^{q_k} h^\Gamma(\gamma) h^Q(v) \left( \sum_{\gamma' \in \Gamma} a_{\gamma'}(q_k;v) h^\Gamma(\gamma') \right) \\
      & = D \cdot H(q_k) \cdot \left( \sum_{\gamma' \in \Gamma} a_{\gamma'}(q_k;v) h^\Gamma(\gamma') \right)\\
      & = D \cdot \sum_{\gamma \in \Gamma} \sum_{v =q_0}^{q_k} h^\Gamma(\gamma) h^Q(v) a_{\gamma}(q_k;v) \\
      & = s_k(\hat{a})
  \end{align*}
  
  Hence $\bm{s}(\tilde{a}) = \bm{s}(\hat{a}) = \bm{s}(a)$. We conclude that $\tilde{a}$ is feasible and attains the same objective as $a$.
\end{proof}

\newpage

\end{document}